\documentclass[11pt]{article}
\usepackage{amsmath,amssymb,amsthm} 
\usepackage{bbm}
\usepackage{stmaryrd}
\usepackage{paralist}
\usepackage{hyperref}
\usepackage{color}

\setdefaultenum{(i)}{}{}{}
\usepackage[margin=1in]{geometry} 
\usepackage[pdftex]{graphicx}
\usepackage{subfigure}
\theoremstyle{plain}
 \numberwithin{mythm}{section}

\DeclareMathAlphabet\scr{U}{scr}{m}{n}
\SetMathAlphabet\scr{bold}{U}{scr}{b}{n}
  \DeclareFontFamily{U}{scr}{\skewchar\font'177}%
  \DeclareFontShape{U}{scr}{m}{n}{<-6>rsfs5<6-8>rsfs7<8->rsfs10}{}%
  \DeclareFontShape{U}{scr}{b}{n}{<-6>rsfs5<6-8>rsfs7<8->rsfs10}{}%

\numberwithin{equation}{section}

\renewenvironment{thebibliography}[1]{%
\begin{oldthebibliography}{#1}%
\setlength{\baselineskip}{.8em}
\linespread{.8}
\small
\setlength{\parskip}{0ex}%
\setlength{\itemsep}{.1em}%
}%
{%
\end{oldthebibliography}%
}

\DeclareMathAlphabet\scr{U}{scr}{m}{n}
\SetMathAlphabet\scr{bold}{U}{scr}{b}{n}
  \DeclareFontFamily{U}{scr}{\skewchar\font'177}%
  \DeclareFontShape{U}{scr}{m}{n}{<-6>rsfs5<6-8>rsfs7<8->rsfs10}{}%
  \DeclareFontShape{U}{scr}{b}{n}{<-6>rsfs5<6-8>rsfs7<8->rsfs10}{}%
\newcommand{\E}{\scr E}

\newcommand{\zu}{]\!]}
\newcommand{\eps}{{\varepsilon}}
\newcommand{\beao}{\begin{eqnarray*}}
\newcommand{\eeao}{\end{eqnarray*}\noindent}
\newcommand{\beam}{\begin{eqnarray}}
\newcommand{\eeam}{\end{eqnarray}\noindent}
\newcommand{\ov}{\overline}
\newcommand{\un}{\underline}
\newcommand{\wh}{\widehat}
\newcommand{\wt}{\widetilde}

\newcommand{\vt}{{\vartheta}}

\def\bbr{{\Bbb R}}   
\def\bbn{{\Bbb N}}

\newtheorem{Satz}{Theorem}[section]

\newtheorem{Lemma}[Satz]{Lemma}
\newtheorem{Definition}[Satz]{Definition}

\newtheorem{Bemerkung}[Satz]{Remark}
\newtheorem{Annahme}[Satz]{Assumption}

\numberwithin{equation}{section}

\begin{document}

\title{Optimal Liquidity Provision\footnote{The authors thank Pierre Collin-Dufresne, Paolo Guasoni, Jan Kallsen, Ren Liu, Mathieu Rosenbaum, and Torsten Sch\"oneborn 
for fruitful discussions and two anonymous referees for valuable comments.}}
\author{
Christoph K\"uhn\thanks{Goethe-Universit\"at Frankfurt, Institut f\"ur Mathematik, D-60054 Frankfurt a.M., Germany, e-mail: \texttt{ckuehn@math.uni-frankfurt.de}.
Financial support by the Deutsche Forschungsgemeinschaft~(DFG), Project ``Optimal Portfolios in Illiquid Financial Markets and in Limit Order Markets'',  
is gratefully acknowledged.}
\and
Johannes Muhle-Karbe\thanks{ETH Z\"urich, Departement f\"ur Mathematik, R\"amistrasse 101, CH-8092, Z\"urich, Switzerland, and Swiss Finance Institute, email \texttt{johannes.muhle-karbe@math.ethz.ch}. Partially supported by the National Centre of Competence in Research ``Financial Valuation and Risk Management'' (NCCR FINRISK), Project D1 (Mathematical Methods in Financial Risk Management), of the Swiss National Science Foundation (SNF).}
}
\maketitle

\begin{abstract}
A small investor provides liquidity at the best bid and ask prices of a limit order market. For small spreads and frequent orders of other market participants, we explicitly determine 
the investor's optimal policy and welfare. In doing so, we allow for general dynamics of the mid price, the spread, and the order flow, as well as for arbitrary preferences of the 
liquidity provider under consideration.
\end{abstract}




\section{Introduction}

Trades on financial markets are instigated by various motives. For example, mutual funds rebalance their portfolios, derivative positions are hedged, and margin calls may necessitate the liquidation of large asset positions. Such trades require counterparties who provide the necessary liquidity to the market. Traditionally, this market making role was played by designated ``specialists'', who agreed on contractual terms to match incoming orders in exchange for earning the spread between their bid and ask prices. As stock markets have become automated, this quasi-monopolistic setup has given way to \emph{limit order markets} on many trading venues. Here, electronic \emph{limit order books} collect all incoming orders, and automatically pair matching buy and sell trades. Such limit order markets allow virtually all market participants to engage in systematic liquidity provision, which has become a popular algorithmic trading strategy for hedge funds. 

The present study analyzes optimal strategies for liquidity provision and their performance. 
In contrast to most previous work on market making, we do \emph{not} consider a 
single \emph{large} monopolistic specialist (e.g., \cite{garman.76,amihud.mendelson.80,ho.stoll.81,avallaneda.stoikov.08,gueant.al.13}) who optimally sets the bid-ask spread. 
Instead, as in \cite{kuehn.stroh.10,cartea.jaimungal.13,gueant.al.13,guilbaud.pham.13,pham.fodra.13}, we focus on a \emph{small} liquidity provider who chooses how much liquidity to provide 
by placing limit buy and sell orders at exogenously given bid and ask prices, respectively. For tractability, we assume that limit orders of the liquidity provider are fully executed against any incoming market order, and, by the above choice of the limit prices, her orders enjoy 
priority over limit orders submitted by other market participants. Thereby, we abstract from incentives to place orders at different limit prices, which
leads to an enormous dimensionality reduction of the strategy space that has to be considered. To wit, we do not have to model the whole order order book. Instead, our model is fully specified by the 
bid-ask price processes and the arrival times of market orders of other market participants. We assume that the mid-price of the risky asset follows a martingale and 
consider the practically relevant limiting regime of small spreads and frequent orders of other market participants. Thereby, we obtain explicit formulas in a general setting 
allowing for arbitrary dynamics of the mid price, the spread, and the order flow, as well as for general preferences of the liquidity provider under 
consideration.\footnote{Related results for models with small trading costs have recently been determined 
by \cite{rosenbaum.tankov.13,martin.12,soner.touzi.13,kallsen.muhlekarbe.13a,kallsen.muhlekarbe.13b}. 
These correspond to optimal trading strategies for liquidity takers, whose demand is matched by liquidity providers such as the ones considered here.}

%%% In summary, we make rather restrictive assumptions on order placement and execution, but allow for considerable flexibility in the modeling of bid-ask prices, order flow, 
%%% and preferences. 
Given the liquidity provider's risk aversion, the asset's volatility, and the arrival rates of exogenous orders, the model tells us how much liquidity to provide by placing limit orders.
However, our model abstracts from the precise microstructure of order books, in particular from the finite price grid and the use of information about order volumes
in the book. In this spirit, we work with diffusion processes that are more tractable than integer-valued jump processes.
Ignoring volume effects, our model carries the flavor of the standard frictionless market model and models with proportional transaction costs.
Consequently, the model does not answer the question whether to place, say, the limit buy order of optimal size {\em exactly} at the current best bid price
or possibly one tick above/below it.

In this setting, 
the optimal policy is determined by an upper and lower boundary for the monetary position in the risky asset, to which the investor trades whenever an exogenous market order of another 
market participant arrives. Hence, these target positions determine the amount of liquidity the investor posts in the limit order book. K\"uhn and Stroh \cite{kuehn.stroh.10} characterize 
these boundaries by the solution of a free boundary problem for a log-investor with unit risk aversion, who only keeps long positions in a market with constant order flow and 
bid-ask prices following geometric Brownian motion with positive drift. In the present study, we show in a general setting with a martingale mid price that -- in the limit for small 
spreads and frequent orders 
of other market participants --  the upper and lower target positions are given explicitly by
\begin{equation}\label{eq:policy}
\overline{\beta}_t= \frac{2\varepsilon_t \alpha^{(2)}_t}{\mathrm{ARA}(x_0)\sigma^2_t}, \quad \underline{\beta}_t= -\frac{2\varepsilon_t \alpha^{(1)}_t }{\mathrm{ARA}(x_0)\sigma^2_t}.
\end{equation}
In these formulas, $2\varepsilon_t$ is the width of the relative bid-ask spread, $\alpha^{(1)}_t$ and $\alpha^{(2)}_t$ are the arrival rates of market sell and buy orders of other market
participants, $\sigma_t$ is the volatility of the risky asset returns, and $\mathrm{ARA}(x_0)$ is the absolute risk aversion of the investor at her initial position $x_0$. To wit, the
optimal amount of liquidity provided is inversely proportional to the inventory risk caused by the asset's local variance, scaled by the investor's risk aversion. Conversely, liquidity
provision is proportional to the compensation per trade (i.e., the relative spread~$2\varepsilon_t$), and the arrival rates $\alpha^{(1)}_t$ respectively $\alpha^{(2)}_t$. The product of these two terms plays the role of the risky asset's expected returns in the usual Merton position, in that it describes the investor's average revenues per unit time, that are traded off against her risk aversion and the variance of the asset returns to determine the optimal target position. Here, however, revenues are derived by netting other traders' buy and sell orders, unlike for the classical Merton problem, where they are generated by participating in trends of the risky asset. Note that the above policy is myopic, in that it only depends on the local dynamics of the model; future variations are not taken into account at the leading order.

The performance of the above strategy can also be quantified. At the leading order, its certainty equivalent is given by
\begin{equation}\label{eq:welfare}
x_0+\frac{\mathrm{ARA}(x_0)}{2}E\left[\int_0^T  (\overline{\beta}^2_t 1_{A^{(1)}_t} + \underline{\beta}^2_t 1_{A^{(2)}_t} )\sigma_t^2 dt \right],
\end{equation}
where $\omega\in A^{(1)}_t$ if the investor's last trade before time $t$ was a purchase so that her position is close to the upper boundary $\overline{\beta}_t$. Likewise, $\omega\in A^{(2)}_t$ if the investor's position is close to the lower boundary $\underline{\beta}_t$ after her last trade was a sale of the risky asset. Hence, the certainty equivalent of providing liquidity in the limit order market is given by the average (with respect to states and business time $\sigma^2_t dt$) of future squared target positions, rescaled by risk aversion.\footnote{This is in direct analogy to the results for models with proportional transaction costs \cite[Equation (3.4)]{kallsen.muhlekarbe.13b}; since the mid price follows a martingale in our model, the marginal pricing measure coincides with the physical probability here.} If all model parameters are constant, the above formula simplifies to
$$
x_0+\frac{(2\varepsilon \alpha^{(1)})(2\varepsilon \alpha^{(2)})}{2\mathrm{ARA}(x_0)\sigma^2}T.
$$
In this case, liquidity provision is therefore equivalent at the leading order to an annuity proportional to the ``drift rates'' $2\varepsilon \alpha^i$ of the investor's revenues from
purchases respectively sales, divided by two times the investor's risk aversion, times the risky asset's variance. In the symmetric case $\alpha^{(1)}=\alpha^{(2)}=\alpha$, this is in direct analogy to
the corresponding result for the classical Merton problem in the Black-Scholes model, where the equivalent annuity is given by the squared Sharpe ratio divided by two times the investor's
risk aversion. For a given total order flow $\alpha^{(1)}+\alpha^{(2)}$, asymmetries $\alpha^{(1)} \neq \alpha^{(2)}$ reduce liquidity providers' certainty equivalents, since they reduce the opportunities to earn the spread with little inventory risk by netting successive buy and sell trades. 

Our model is an overly optimistic playground for liquidity providers. These do not incur monitoring costs 
and always achieve full execution of their limit orders without having to 
further narrow the spread.\footnote{Partial execution of limit orders is studied by Guilbaud and Pham \cite{guilbaud.pham.13b}.  A model where liquidity providers have to narrow the 
spread by a discrete ``tick'' to gain execution priority is analyzed in \cite{guilbaud.pham.13}.}   
Submission and deletion of orders is free.
Moreover, since market orders of other market participants do not move the current best bid or ask prices, they earn the full spread between alternating buy and sell trades, only subject to the risk of intermediate price changes. This changes substantially if market prices systematically rise respectively fall for purchases respectively sales of other market participants, as acknowledged in the voluminous literature on price impact (e.g., \cite{bertsimas.lo.98,almgren.chriss.01,obizhaeva.wang.13}). These effects can stem, e.g., from adverse selection, as informed traders prey on the liquidity providers 
\cite{glosten.milgrom.85}, or from large incoming orders that eat into the order book \cite{obizhaeva.wang.13}. Our model can be extended to account for price impact of incoming orders 
equal to a fraction $\kappa\in[0,1)$ of the half-spread.\footnote{If the price impact is almost the half-spread, this leads to a model similar to the one of Madhavan et al.\ \cite{madhavan.al.97}.} 
This extension is still tractable; indeed, the above formula \eqref{eq:policy} for the leading-order optimal position limits generalizes to
\begin{equation}\label{23.4.2014.1}
\overline{\beta}_t = \frac{2\varepsilon_t ((1-\frac{\kappa}{2})\alpha^{(2)}_t-\frac{\kappa}{2}\alpha^{(1)}_t)}{\mathrm{ARA}(x_0) \sigma^2_t}, \quad \underline{\beta}_t = -\frac{2\varepsilon_t ((1-\frac{\kappa}{2})\alpha^{(1)}_t-\frac{\kappa}{2}\alpha^{(2)}_t)}{\mathrm{ARA}(x_0) \sigma^2_t}.
\end{equation}
For a symmetric order flow ($\alpha^{(1)}_t=\alpha^{(2)}_t=\alpha_t$), these formulas reduce to
\begin{equation*}
\overline{\beta}_t = \frac{2\varepsilon_t (1-\kappa)\alpha_t}{\mathrm{ARA}(x_0) \sigma^2_t}, \quad \underline{\beta}_t = -\frac{2\varepsilon_t (1-\kappa)\alpha_t}{\mathrm{ARA}(x_0) \sigma^2_t}.
\end{equation*}
That is, liquidity provision is simply reduced by a factor of $1-\kappa$ in this case. Here, the intuition is that, for $\kappa \approx 1$, price impact almost neutralizes the proportional transaction cost $\varepsilon_t$ the liquidity provider earns per trade. Hence, market making becomes unprofitable in this case as dwindling earnings are outweighed by inventory risk.

Even with price impact, the leading-order optimal certainty equivalent is still given by \eqref{eq:welfare}, if one replaces the trading boundaries accordingly. 
Hence, liquidity providers' profits shrink as they reduce their position limits due to price impact.
In view of the still optimistic assumptions, the optimal strategy associated with (\ref{23.4.2014.1}) can at least serve as an upper bound (respectively a lower bound 
if it takes a negative value) for a strategy to follow in practise. It can be computed explicitly and is easy to interpret. 

The remainder of the article is organized as follows. Our model is introduced in Section 2. Subsequently, the main results of the paper are presented in Section 3, and 
proved in Section 4. Finally, Section 5 extends the model to allow for price impact of incoming orders.

\section{Model}

\subsection{Limit Order Market}

We consider a financial market with one safe asset, normalized to one, and one risky asset, which can be traded either by \emph{market orders} or by \emph{limit orders}. 
Market orders are executed immediately, but purchases at time~$t$ are settled at a higher ask price~$(1+\varepsilon_t)S_t$, whereas sales only earn a lower bid 
price~$(1-\varepsilon_t)S_t$.\footnote{That is, $\varepsilon_t$ is the halfwidth of the relative bid-ask spread.} In contrast, limit orders can be put into the order book with an 
arbitrary exercise price, but are only executed once a matching order of another market participant arrives. Handling the complexity of limit orders with arbitrary exercise prices 
is a daunting task. To obtain a tractable model, we therefore follow \cite{kuehn.stroh.10} in assuming that limit buy or sell orders can only be placed at the 
current best bid or ask price, respectively. This can be justified as follows for \emph{small} investors, whose orders do not move market prices, and
for continuous best bid and ask prices. In this case, placing (and constantly 
updating) limit buy orders at a ``marginally'' higher price than the current best-bid price~$(1-\varepsilon_t)S_t$ guarantees execution as soon as the next market sell order of another 
trader arrives. For the sake of tractability, we abstract from the presence of a finite tick size.
Consequently, limit buy orders with a higher exercise price are executed at the same time but at a higher cost, 
whereas, by continuity, exercise prices below the current best bid are only executed later. This argument implicitly assumes that the incoming orders of other market participants are liquidity-driven 
and small, so that they do not move market prices (we show how to relax this assumption in 
Section~\ref{sec:impact}). Moreover, the investor under consideration is even smaller, in that her orders also don't influence market prices and are executed immediately against 
any incoming order of another market participant.\footnote{Partial execution is studied by Guilbaud and Pham \cite{guilbaud.pham.13b}.}  These assumptions greatly reduce the complexity 
of the problem. Yet, the model still retains the key tradeoff between making profits by providing liquidity, and the inventory risk caused by the positions built up along the way.\\
%%% Even the occurence of partial executions or priority of other liquidity providers, from which we abstract, may be taken into account implicitly by choosing smaller arrival rates
%%% of the exogenous market orders in our model. 
 
Let us now formalize trading in this limit order market. All stochastic quantities are defined on a filtered probability space $(\Omega,\scr{F},(\scr{F}_t)_{t\in[0,T]},P)$ 
satisfying the usual conditions. Strategies are described by quadruples $\mathfrak{S}=(M_t^B,M_t^S,L_t^B,L_t^S)_{t \in [0,T]}$ of predictable processes. Here, the nondecreasing 
processes $M^B_t$ and $M^S_t$ represent the investor's cumulated market buy and  sell orders until time $t$, respectively. 
$M_t^B$ and $M_t^S$  possess left and right limits, but may have double jumps. 
% Therefore, we need the following definition

For a c\`adl\`ag process~$Y_t$ and a process~$M_t$ of finite variation, the integral of $Y_t$ with respect to $M_t$ is defined as 

\beam\label{20.1.2015.2}
\int_0^t (Y_{s-},Y_s)\,dM_s := \int_0^t Y_{s-}\,d M^r_s + \sum_{0\le s<t} Y_s (M_{s+}-M_s),
\eeam
where the integrator $M^r_t:=M_t-\sum_{0\le s<t}(M_{s+}-M_s)$ is c\`adl\`ag, i.e., the first term on the right-hand side of (\ref{20.1.2015.2}) is just a standard Lebesgue-Stieltjes
integral, see also \cite[Equation (2.2)]{kuehn.stroh.10}. For continuous integrands~$Y_t$, 
we set 
\beam\label{23.1.2015.1}
\int_0^t Y_s\,dM_s := \int_0^t (Y_{s-},Y_s)\,dM_s.
\eeam
%
% \int_0^t Y_s\,d M^r_s + \sum_{0\le s<t} Y_s (M_{s+}-M_s)$, where the integrator $M^r_t:=M_t-\sum_{0\le s<t}(M_{s+}-M_s)$ is c\`adl\`ag, see also 
% \cite{kuehn.stroh.10}.

$L^B_t$ (respectively $L^S_t$) specifies the size of the limit buy order 
with limit price $(1-\eps_t)S_t$ (respectively the limit sell order with limit price $(1+\eps_t)S_t$) in the book at time $t$, i.e., the amount that is bought or sold if an exogenous market 
sell or buy order arrives at time $t$.\footnote{The assumption that $L_t^B$ and $L_t^S$ can be arbitrary predictable processes is justified because the submission and deletion 
of limit orders is typically free.} Fix an initial position of $x_0$ units in the safe and $x=0$ units in the risky asset. The number of risky assets, denoted by $\varphi_t$, changes 
by market orders, and when limit buy or sell orders are executed at the jump times of some counting processes $N^{(1)}_t$ or $N^{(2)}_t$, respectively. At the jump times of $N_t^{(1)}$, the sell 
order of 
another market participant arrives so that the risky position of the liquidity provider is increased according to the number of corresponding limit orders in the book, and analogously 
for incoming buy orders at the jump times of $N_t^{(2)}$. Market orders are executed at the less favorable side of the bid-ask spread, whereas limit orders 
are matched against other traders' orders at the more favorable side. This leads to the following definition.
\begin{Definition}
A pair $(\varphi_t^0,\varphi_t)_{t \in [0,T]}$ specifying the number of monetary units and risky assets that the investor holds is called a \emph{self-financing portfolio process}
iff it can be written as 
\begin{equation}\label{eq:sf1}
\varphi_t=M^B_t-M^S_t+\int_0^{t-}L^B_s\,dN^{(1)}_s-\int_0^{t-}L^S_s\,dN^{(2)}_s
\end{equation}
and
\begin{align}
\varphi^0_t=x_0 &-\int_0^t (1+\eps_s)S_s\,dM^B_s+\int_0^t (1-\eps_s)S_s\,dM^S_s\label{eq:sf2}\\
&-\int_0^{t-}L^B_s(1-\eps_s)S_s\,dN^{(1)}_s+\int_0^{t-}L^S_s(1+\eps_s)S_s\,dN^{(2)}_s\nonumber
\end{align}
for some strategy~$\mathfrak{S}=(M_t^B,M_t^S,L_t^B,L_t^S)_{t \in [0,T]}$, where the integrals with respect to $M^B_t$ and $M^S_t$ are defined in (\ref{23.1.2015.1}). 
\end{Definition}

\begin{Bemerkung}
In reality, liquidity providers do not always have priority over other traders, so that their limit orders are not executed against all matching market orders. 
In our model, one can account for this by choosing smaller arrival rates for $N^{(1)}_t$ and $N^{(2)}_t$. To wit, these processes may just count the part of the market orders 
that actually trigger an execution.

Partial executions, from which we abstract, can to some extent be taken into account by reducing the arrival rates. Indeed, in the limiting regime of frequently 
arriving market orders, partial execution has a similar effect as a full execution that takes place only with some probability.
\end{Bemerkung}

Let now specify the primitives of our model. We work in a general It\^o process setting; in particular, no Markovian structure is required. The mid price follows 
$$
dS_t = S_t\sigma_t\,dW_t,\quad S_0>0, 
$$
for a Brownian motion $W_t$ and a volatility process $\sigma_t$. Assuming the mid-price of the risky asset to be a martingale allows to disentangle the effects of liquidity provision from pure investment due to trends in the risky asset; on a technical level, it is also needed to obtain both long and short positions even in the limit for small spreads. This assumption is reasonable since ``market making is typically not directional, in the sense that it does not profit from security prices going up or down'' \cite{guilbaud.pham.13}.  Moreover, as in the optimal execution literature (e.g., \cite{bertsimas.lo.98,almgren.chriss.01, obizhaeva.wang.13}), it is also justified by the time scales under consideration: we are \emph{not} dealing with long-term investment here, but much rather focusing on high-frequency liquidity provision strategies which are typically liquidated and evaluated at the end of a trading day \cite{menkveld.12}. Models for high-frequency strategies designed to profit from the predictability of short-term drifts are studied in \cite{cartea.al.11,guilbaud.pham.13b}.

 The arrival times of sell and buy orders by other market participants are modeled by counting processes $N_t^{(1)}$ and $N_t^{(2)}$ with absolutely continuous jump intensities 
 $\alpha^{(1)}_t$ and $\alpha^{(2)}_t$, respectively;\footnote{That is, $\alpha^i_t$ are predictable processes and $N^i_t - \int_0^t \alpha^i_s\,ds$ are local martingales for $i=1,2$.} 
 we assume that $N^{(1)}_t$ and $N^{(2)}_t$ a.s.\ never jump at the same time.
 In contrast to most of the previous literature, we do not restrict ourselves to Poisson processes with independent and identically distributed inter-arrival times. Instead, we allow for general arrival rates, thereby recapturing uncertainty about future levels and also empirical observations such as the U-shaped distribution of order flow over the trading day.
 
 We are interested in limiting results for a small relative half-spread $\varepsilon_t$. Therefore, we parametrize it as 
 $$\varepsilon_t=\varepsilon \mathcal{E}_t,$$
  for a small parameter $\varepsilon$ and an It\^o process $\mathcal{E}_t$. Unlike for models with proportional transaction costs (e.g., \cite{whalley.wilmott.97,janecek.shreve.04}), where it is natural to assume that all other model parameters remain constant as the spread tends to zero, the width of the spread is inextricably linked to the arrival rates of exogenous market orders here. Indeed, market orders naturally occur more frequently for more liquid markets with smaller spreads. Hence, we rescale the arrival rates  accordingly:
\beao
\alpha^{(1)}_t = \lambda^{(1)}_t \eps^{-\vt},\quad \alpha^{(2)}_t = \lambda^{(2)}_t \eps^{-\vt},\quad\mbox{for some $\vartheta\in(0,1)$.}
\eeao 

Here, $\vartheta>0$ ensures that the arrival rate of exogenous market orders explodes
as the bid-ask spread vanishes for $\varepsilon \to 0$. Nevertheless, the risk that limit orders are not executed
fast enough is a crucial factor for the solution in the limiting regime.
$\vartheta<1$ is assumed to ensure that the profits from liquidity
provision vanish as $\varepsilon \to 0$.
% Otherwise the myopic structure of the optimal solution would not hold for arbitrary utility functions and price processes with stochastically
% dependent increments. 
% Of course the results are asymptotic expansions and vanishing profits need not be understood literally. 
Higher arrival rates necessitate extensions of the model such as a price impact of incoming orders; see Section~\ref{sec:impact} for more details. 
In our optimal policy and the corresponding utility, the exponent $\vartheta$ only appears in the rates of the asymptotic expansions; the leading-order terms are fully determined 
by the arrival rates $\alpha^{(1)}_t,\alpha^{(2)}_t$.

The processes $\lambda^{(1)}_t$, $\lambda^{(2)}_t$, $\sigma_t$, and $\mathcal{E}_t$ satisfy the following technical assumptions:

\begin{Annahme}\label{22.5.2013.1}
$\lambda^{(1)}_t,\lambda^{(2)}_t,\sigma^2_t$, and $\mathcal{E}_t$ are positive continuous processes that are bounded and bounded away from zero. 
Furthermore, $\mathcal{E}_t$ is a semimartingale. Its predictable finite variation part and the quadratic variation process of its local martingale part are absolutely continuous 
with a bounded rate.
\end{Annahme}

Note that we allow for any stochastic dependence of the processes $\lambda^i_t$ and $\mathcal{E}_t$. In the market microstructure literature (e.g., \cite{cvitanic.kirilenko.2010}), plausible distributions of trading times as functions of the current bid-ask prices are derived.

\subsection{Preferences}

The investor's preferences are described by a general von Neumann-Morgenstern utility function $U: \mathbb{R} \to \mathbb{R}$ satisfying the following mild regularity conditions:

\begin{Annahme}\label{20.8.2013.2}
\begin{enumerate}
\item $U$ is strictly concave, strictly increasing, and twice continuously differentiable. 
\item The corresponding absolute risk aversion is bounded and bounded away from zero:
\begin{equation}\label{eq:ARA}
c_1 < \mathrm{ARA}(x):=-\frac{U''(x)}{U'(x)} < c_2, \quad \forall x\in\bbr,
\end{equation}
for some constants $c_1,c_2>0$.
\end{enumerate}
\end{Annahme}

\begin{Bemerkung}
Since $U'(x)=U'(0)\exp(\int_0^x U''(y)/U'(y)\,dy)$, Condition (\ref{eq:ARA}) implies that
\begin{equation}\label{eq:marginal}
U'(x), |U''(x)| \leq C \exp(-c_2 x),\quad \forall x\le 0\quad\mbox{and}\quad U'(x), |U''(x)| \leq C \exp(-c_1 x),\quad \forall x>0,
\end{equation}
for some constant $C>0$.
\end{Bemerkung}

The arch-example satisfying these assumptions is of course the exponential utility $U(x)=-\exp(-c x)$ with constant absolute risk aversion $c>0$. 
Analogues of our results can also be obtained for utilities defined on the positive half line, such as power utilities with constant \emph{relative} risk aversion. 
Here, we focus on utilities whose absolute risk aversion is uniformly bounded, because these naturally lead to bounded monetary investments in the risky asset, 
in line with the ``risk budgets'' often allocated in practice:

\begin{Definition}\label{21.5.2013.1}
A family of self-financing portfolio processes~$(\varphi^{0,\eps},\varphi^\eps)_{\eps\in(0,1)}$ in the limit order market is called \emph{admissible} if the monetary 
position~$\varphi^\eps S$ held in the risky asset is uniformly bounded.
\end{Definition}

This notion of admissibility is not restrictive. Indeed, it turns out that the optimal positions held in the risky asset 
even converge to zero uniformly as $\eps\to 0$ (cf.\ Theorem \ref{main_theorem}).

\section{Main Results}\label{sec:main}

The main results of the present study are a trading policy that is optimal at the leading order $\varepsilon^{2(1-\vartheta)}$ for 
small relative half-spreads $\varepsilon_t=\varepsilon \mathcal{E}_t$, and an explicit formula for the utility that can be obtained by applying it. To this end, define the monetary trading boundaries 
$$\overline{\beta}_t= \frac{2\varepsilon_t \alpha^{(2)}_t}{\mathrm{ARA}(x_0)\sigma^2_t}, \quad \underline{\beta}_t= -\frac{2\varepsilon_t \alpha^{(1)}_t }{\mathrm{ARA}(x_0)\sigma^2_t},$$
and consider the strategy that keeps the risky position $\beta^\varepsilon_t=\varphi_t S_t$ in the interval $[\underline{\beta}_t,\overline{\beta}_t]$ by means of market orders, 
while constantly updating the corresponding limit orders so as to trade to $\underline{\beta}_t$ respectively $\overline{\beta}_t$ whenever the buy respectively sell order of another market 
participant allows to sell or buy at favorable prices, respectively. Formally, this means that the process $(\beta^\varepsilon_t)_{t \in [0,T]}$ is defined as the 
unique solution to the Skorokhod stochastic differential equation
\beam\label{label1}
d\beta^\eps_{t+} = \beta^\eps_t\sigma_t dW_t + (\ov{\beta}_t-\beta^\eps_t)dN^{(1)}_t + (\un{\beta}-\beta^\eps_t) dN^{(2)}_t + d\Psi_t,\quad \beta^\eps_0=0,
\eeam
where $\Psi_t$ is the minimal finite variation process that keeps the solution in $[\underline{\beta}_t,\overline{\beta}_t]$.\footnote{That is, $\un{\beta}_t \le \beta_t^\eps \le
\ov{\beta}_t$ and $\Psi_t$ is a continuous process of finite variation such that
$\int_0^t 1_{\{\beta_s^\eps = \un{\beta}_s\}}d\Psi_s$ is nondecreasing, $\int_0^t 1_{\{\beta_s^\eps =  \ov{\beta}_s\}}d\Psi_s=0$ is nonincreasing, and
$\int_0^t 1_{\{\un{\beta}_s < \beta_s^\eps <  \ov{\beta}_s\}}d\Psi_s$ vanishes. 
Existence and uniqueness of the solution is guaranteed by Theorem~3.3 in 
S\l omi\'nski and Wojciechowski~\cite{slominski.wojciechowski.2013}, applied to the evolution of (\ref{label1}) between the jump times, without the integrals with respect to $N^{(1)}_t$ and
$N^{(2)}_t$. For the special reflected SDE we consider here, 
the solution is constructed explicitly in \eqref{18.5.2013.1}.} 
This corresponds to the strategy 
\beam\label{20.8.2013.1}
M^B_t := \int_0^t \frac1{S_s}d\Psi_s^+,\quad M^S_t := \int_0^t \frac1{S_s}d\Psi_s^-,\quad L^B_t:=\frac{\ov{\beta}_t-\beta_t^\eps}{S_t},\quad L_t^S:=\frac{\beta_t^\eps-\un{\beta}_t}{S_t}.
\eeam 
The family of associated portfolio processes given by
\beao
d\varphi^\eps_{t+} & = & \frac1{S_t} d\Psi^+_t -  \frac1{S_t} d\Psi^-_t + \frac{\ov{\beta}_t-\beta^\eps_t}{S_t}dN^{(1)}_t + \frac{\un{\beta}_t-\beta^\eps_t}{S_t}dN^{(2)}_t,\quad \varphi^\eps_0=0,\\ 
d\varphi^{0,\eps}_{t+} & = & -(1+\eps_t) d\Psi^+_t + (1-\eps_t) d\Psi^-_t + (1-\eps_t)(\beta^\eps_t - \ov{\beta}_t)dN^{(1)}_t + (1+\eps_t)(\beta^\eps_t - \un{\beta}_t) dN^{(2)}_t,\quad 
\eeao
$\varphi^{0,\eps}_0=x_0$, is admissible with \emph{liquidation wealth processes} $\wh{X}^\eps_t :=\varphi^{0,\eps}_t 
+ \varphi^\eps_t 1_{\{\varphi^\eps_t\ge 0\}} (1-\eps_t)S_t + \varphi^\eps_t 1_{\{\varphi^\eps_t< 0\}} (1+\eps_t)S_t$. The following is the main result of the present paper:

\begin{Satz}\label{main_theorem}
Suppose Assumptions \ref{22.5.2013.1} and \ref{20.8.2013.2} hold. Then, the above policy is optimal at the leading order $\varepsilon^{2(1-\vt)}$, in that:
$$
E[U(\wh{X}^\eps_T)] = U\left(x_0 + \frac{\eps^{2(1-\vt)}}{ARA(x_0)}E\left[\int_0^T\left(\frac{2\mathcal{E}_t^2 (\lambda^{(2)}_t)^2}{\sigma^2_t} 1_{A^{(1)}_t}
+ \frac{2\mathcal{E}_t^2 (\lambda^{(1)}_t)^2}{\sigma^2_t} 1_{A^{(2)}_t}\right)\,dt\right]\right) 
+ o( \eps^{2(1-\vt)}),\quad 
$$
$\eps\to 0$, for the corresponding liquidation wealth processes $(\wh{X}^\eps)_{\eps\in(0,1)}$, whereas
$$
E[U(X^\eps_T)] \le U\left(x_0 + \frac{\eps^{2(1-\vt)}}{ARA(x_0)}E\left[\int_0^T\left(\frac{2\mathcal{E}_t^2 (\lambda^{(2)}_t)^2}{\sigma^2_t} 1_{A^{(1)}_t}
+ \frac{2\mathcal{E}_t^2 (\lambda^{(1)}_t)^2}{\sigma^2_t} 1_{A^{(2)}_t}\right)\,dt\right]\right) 
+ o( \eps^{2(1-\vt)}),\quad 
$$
$\eps\to 0$, for any competing family $(X^\eps)_{\eps\in(0,1)}$ of admissible liquidation wealth processes. 

Here, $\omega\in A^{(1)}_t$ respectively $\omega\in A^{(2)}_t$ means that the investor's last trade before time $t$ was a purchase or sale, respectively, i.e.\ $A^i_t=\{\omega\ |\ (\omega,t) \in A^i\}$, $i=1,2$ for the predictable sets 
\begin{equation}\label{14.5.2013.1}
\begin{split}
A^{(1)} & =  \left\{(\omega,t)\ |\ \sup\{s\in (0,t)\ |\ \Delta N^{(1)}_s>0\} > \sup\{s\in (0,t)\ |\ \Delta N^{(2)}_s>0\}\right\},\\
A^{(2)} & =  \left\{(\omega,t)\ |\ \sup\{s\in (0,t)\ |\ \Delta N^{(2)}_s>0\} \ge \sup\{s\in (0,t)\ |\ \Delta N^{(1)}_s>0\}\right\}.
\end{split}
\end{equation}
(By convention, before the first jump of $(N^{(1)}, N^{(2)})$ all time points belong to $A^{(2)}$).
\end{Satz}

If the model parameters $\lambda^i_t,\sigma_t,\varepsilon_t$ are all constant, then the above formula reduces to
\beao
E[U(\wh{X}^\eps_T)]= U\left(x_0+\frac{2\lambda^{(1)} \lambda^{(2)}}{\mathrm{ARA}(x_0) \sigma^2}T\varepsilon^{2(1-\vartheta)}
\right)+o(\varepsilon^{2(1-\vartheta)}),\quad \eps\to 0.
\eeao

\section{Proofs}\label{sec:proofs}

This section contains the proof of our main result, Theorem \ref{main_theorem}. We proceed as follows: first, it is shown that as the 
relative half-spread $\varepsilon_t=\varepsilon \mathcal{E}_t$ tends to zero and jumps to the trading boundaries $\underline{\beta}_t,\overline{\beta}_t$ become more 
and more frequent for our policy $\beta_t^\varepsilon$, almost all time is eventually spent near $\underline{\beta}_t,\overline{\beta}_t$. Motivated by this result, 
we then construct a frictionless ``shadow market'', which is at least as favorable as the original limit order market, and for which the policy  
that oscillates between $\underline{\beta}_t$ and $\overline{\beta}_t$ is optimal at the leading order for small spreads. In a third step, we then show that the 
utility obtained from applying our original policy $\beta^\varepsilon_t$ matches the one for the approximate optimizer in the more favorable frictionless shadow market at 
the leading order for small spreads, so that our candidate $\beta^\varepsilon_t$ is indeed optimal at the leading order. 

\subsection{An Approximation Result}

As described above, we start by showing that our policy $\beta^\varepsilon_t$ spends almost all time near the boundaries $\underline{\beta}_t,\overline{\beta}_t$ as the relative half-spread $\varepsilon_t=\varepsilon \mathcal{E}_t$ collapses to zero and orders of other market participants become more and more frequent:

\begin{Lemma}\label{lemma_boundary}
On the stochastic interval $\zu \inf\{t>0\ |\ \Delta N^{(1)}_t>0\quad\mbox{or}\quad \Delta N^{(2)}_t>0\},T\zu$, the process
\beao
\left(\beta_t^\eps - \ov{\beta}_t 1_{A^{(1)}_t} - \un{\beta}_t 1_{A^{(2)}_t}\right)\eps^{\vt-1}
\eeao
converges to $0$ uniformly in probability for $\eps\to 0$.
\end{Lemma}

\begin{proof}
The solution of the Skorokhod SDE~\eqref{label1} can be constructed explicitly. Let $(\tau_i^\eps)_{i\in\bbn}$ be the jump times of $N^{(1)}_t$, i.e.\ the jumps 
of $\beta^\varepsilon_t$ to the upper boundary $\overline{\beta}_t$. (To ease notation we suitably extend the model beyond $T$.) 
From $\tau_i^\eps$ up to the next jump time of $(N_t^{(1)},N_t^{(2)})$, the solution
is then given by
\beam\label{18.5.2013.1}
\beta^\eps_t = \exp\left(\int_{\tau_i^\eps}^t \sigma_u\,dW_u - \frac12 \int_{\tau_i^\eps}^t \sigma^2_u\,du 
- \sup\left\{ \int_{\tau_i^\eps}^s \sigma_u\,dW_u - \frac12 \int_{\tau_i^\eps}^s \sigma^2_u\,du - \ln\left(\ov{\beta}_s\right)\ |\ 
s\in[\tau^\eps_i,t]\right\}\right)
\eeam
(analogously after jump times of $N_t^{(2)}$), and $\Psi_t = \Psi_{\tau^\eps_i} + \beta^\varepsilon_t - \ov{\beta}_{\tau^\eps_i} - \int_{\tau_i^\eps}^t \beta^\eps_u\sigma_u\,dW_u$.
Indeed, $\beta^\eps$ from (\ref{18.5.2013.1}) satisfies $d\beta^\eps_t = \beta^\eps_t \sigma_t\,dW_t  
- \beta^\eps_t\,d\left(\sup\left\{ \int_{\tau_i^\eps}^s \sigma_u\,dW_u - \frac12 \int_{\tau_i^\eps}^s \sigma^2_u\,du - \ln\left(\ov{\beta}_s\right)\ |\ 
s\in[\tau^\eps_i,t]\right\}\right)$, and the latter integrator is nondecreasing and on the set~$\{\beta^\eps<\ov{\beta}\}$ even constant
because the above supremum is not attained at $t$ if $\ln(\beta^\eps_t)<\ln(\ov{\beta}_t)$.\\

Define the process
\beao
Y_t:=\int_0^t \sigma_u\,dW_u - \frac12 \int_0^t \sigma^2_u\,du 
- \ln\left(\frac{2\mathcal{E}_t \lambda^{(2)}_t}{\mathrm{ARA}(x_0)\sigma^2_t}\right),\quad t\ge 0.
\eeao
$Y_t$ does not depend on the scaling parameter $\eps$, and, by Assumption~\ref{22.5.2013.1}, it possesses almost surely continuous paths, which implies that
\beam\label{1.9.2014.1}
\sup_{t_1,t_2\in[0,T],\ |t_2-t_1|\le h} |Y_{t_2}-Y_{t_1}|\to 0,\quad\mbox{a.s.},\quad h\to 0. 
\eeam
By (\ref{18.5.2013.1}), one has
\beam\label{2.9.2014.1}
\frac{\beta^\eps_t}{\ov{\beta}_t} =  \exp\left( Y_t  -  \sup_{s\in[\tau^\eps_i,t]} Y_s\right)
\eeam
for all $t$ between $\tau^\eps_i$ and the next jump time of $(N^{(1)}_t,N^{(2)}_t)$ and 
\beam\label{2.9.2014.1b}
\exp\left( Y_t  -  \sup_{s\in[\tau^\eps_i,t]} Y_s\right)
\ge \exp\left(-\sup_{t_1,t_2\in[0,T],\ |t_2-t_1|\le h} |Y_{t_2}-Y_{t_1}|\right)\mbox{for}\ h>0, t\in(\tau^\eps_i,(\tau^\eps_i + h)\wedge T].
\eeam
% In addition, we have (nach Rechnungen im aktuellen Beweis)
% \beao
% \max_{i=1,\ldots,\lfloor 2 T\lambda^{(1)}_{\max}\eps^{-\vt}\rfloor} \left(\tau^\eps_{i+1}-\tau^\eps_i\right) \to 0\quad\mbox{in probability for}\ \eps\to 0.
% \eeao

Now, fix any $\widetilde{\varepsilon}>0$. By (\ref{1.9.2014.1}), there exists $\wt{h}>0$ with
\beam\label{2.9.2014.2}
P\left(\exp\left(-\sup_{t_1,t_2\in[0,T],\ |t_2-t_1|\le \wt{h}} |Y_{t_2}-Y_{t_1}|\right)<1-\wt{\eps}\right)\le \frac{\wt{\eps}}3.
\eeam
Note that, after a limit {\em sell} order execution, $\beta_t^\eps$ jumps to $\un{\beta}_t<0$ and then cannot enter the region~$[0,(1-\wt{\eps})\ov{\beta}_t)$ before 
the next limit \emph{buy} order execution. As a result, we can use (\ref{2.9.2014.1}) and
estimate the excursions away from the upper trading boundary $\overline{\beta}_t$ as follows:
\begin{equation}\label{eq:prob}
P\left(\exists t\in(\tau^\eps_1,T]\mbox{\ s.t.\ }\frac{\beta^\eps_t}{\ov{\beta}_t}\in [0,1-\wt{\eps})\right)
 \le  P\left(M_{1,\eps}\cup\bigcup_{i=1}^{\lfloor 2 T\lambda^{(1)}_{\max}\eps^{-\vt}\rfloor} M_{2,\eps,i}\cup \bigcup_{i=1}^{\lfloor 2 T\lambda^{(1)}_{\max}\eps^{-\vt}\rfloor} M_{3,\eps,i}\right),
\end{equation}
where
\beao
& & M_{1,\eps}:=\{\omega\in\Omega\ |\ N^{(1)}_T(\omega)>\lfloor 2 T\lambda^{(1)}_{\max}\eps^{-\vt}\rfloor\},\\ 
& & M_{2,\eps, i}:=\{\omega\in\Omega\ |\ \tau^\eps_{i+1}(\omega) - \tau^\eps_i(\omega) > \wt{h}\},\\
& & M_{3,\eps, i}:=\left\{\omega\in\Omega\ |\ \exp( Y_t(\omega)  -  \sup_{s\in[\tau^\eps_i(\omega),t]} Y_s(\omega))<1-\wt{\eps}\ \mbox{for 
some\ }t\in(\tau^\eps_i(\omega),(\tau^\eps_i(\omega)+\wt{h})\wedge T]\right\},
\eeao
with $\lambda^{(1)}_{\max}\eps^{-\vt}$ being an upper bound for the jump intensity of the counting process~$N^{(1)}_t$. 
In plain English, there are either many jumps to the upper boundary $\overline{\beta}_t$, and/or there is a long-time excursion away from $\overline{\beta}_t$, and/or there is a 
short excursion that nevertheless takes the risky position $\beta^\varepsilon_t$ sufficiently far way from the boundary $\overline{\beta}_t$. In the sequel, we show that the probability 
for these events is smaller than $\widetilde{\varepsilon}$ for $\varepsilon$ sufficiently small. Observe that after the first jump of $(N^{(1)}_t,N^{(2)}_t)$ we have 
$0<\beta^\eps_t \le \ov{\beta}_t$ on $A^{(1)}_t$ and $\un{\beta}_t \le \beta^\eps_t<0$ on $A^{(2)}_t$. As 
$\ov{\beta}_t=\frac{2\mathcal{E}_t\lambda^{(2)}_t}{\sigma^2_t}\eps^{1-\vt}$ and the process $\frac{2\mathcal{E}_t\lambda_t^{(2)}}{\sigma_t^2}$ 
is bounded, the estimate for \eqref{eq:prob} in turn yields that $|\beta_t^\eps-\ov{\beta}_t|\eps^{\vt-1}1_{A_t^{(1)}}\to 0$ uniformly in probability. 
By applying the same arguments to $\beta_t^\eps$ on $A_t^{(2)}$ we obtain the assertion. 

Let us now derive the required estimates for \eqref{eq:prob}. First, recall that the time-changed process $u\mapsto N^{(1)}_{\Gamma_u}$ with 
$\Gamma_u:=\inf\{v\ge 0\ |\ \int_0^v \lambda^{(1)}_s\eps^{-\vt}\,ds= u\}$ is a standard Poisson process (cf.\ \cite[Theorem~16]{bremaud.81}) and $\Gamma_{\lambda^{(1)}_{\max}\eps^{-\vt}T}\ge T$. 
As a result:
\beam\label{label14}
P(M_{1,\eps}) \le P(N^{(1)}_{\Gamma_{\lambda^{(1)}_{\max}\eps^{-\vt}T}} >\lfloor 2 T\lambda^{(1)}_{\max}\eps^{-\vt}\rfloor)\le \frac{\wt{\eps}}3,\quad\mbox{for $\eps$ small enough}, 
\eeam
by the law of large numbers.

Next, since the jump intensity of $N^{(1)}_t$ is bounded from below by some $\lambda^{(1)}_{\min}\eps^{-\vt}>0$, we obtain
$$
P\left(\tau^\eps_{i+1}-\tau^\eps_i > \frac{x}{\lambda^{(1)}_{\min} \eps^{-\vt}}\right)\le \exp(-x),\quad x\in\bbr_+,\ i\in\bbn. 
$$
Choosing $x=\lambda^{(1)}_{\min}\eps^{-\vt}\wt{h}$, this estimate yields
$$
P\left( \tau^\eps_{i+1}-\tau^\eps_i > \wt{h}\quad \mbox{for some\ } i=1,\ldots,\lfloor 2T\lambda^{(1)}_{\max}\eps^{-\vt}\rfloor\right)
\le \lfloor 2 T\lambda^{(1)}_{\max} \eps^{-\vt}\rfloor \exp\left(-\lambda^{(1)}_{\min}\eps^{-\vt}\wt{h} \right).
$$
This in turn gives
\beam\label{label13}
P\left(\bigcup_{i=1}^{\lfloor 2 T\lambda^{(1)}_{\max}\eps^{-\vt}\rfloor} M_{2,\eps,i}\right) \le \frac{\wt{\eps}}3,\quad\mbox{for $\eps$ small enough.}
\eeam

Finally, by (\ref{2.9.2014.1b}) and (\ref{2.9.2014.2}), we have that
\beam\label{label12}
P\left(\bigcup_{i=1}^{\lfloor 2 T\lambda^{(1)}_{\max}\eps^{-\vt}\rfloor}M_{3,\eps, i}\right) 
\le P\left( \exp\left(-\sup_{t_1,t_2\in[0,T],\ |t_2-t_1|\le \wt{h}} |Y_{t_2}-Y_{t_1}|\right)<1-\wt{\eps}\right)
\le \frac{\wt{\eps}}3
\eeam
for $\eps$ small enough. Piecing together \eqref{label14}, \eqref{label13}, and \eqref{label12}, it follows that the probability in \eqref{eq:prob} is indeed bounded by 
$\widetilde{\varepsilon}$ 
for sufficiently small $\varepsilon$. This completes the proof.
\end{proof}

\subsection{An Auxiliary Frictionless Shadow Market}\label{sec:proof1}

Similarly as for markets with proportional transaction costs \cite{kallsen.muhlekarbe.10} and for limit order markets \cite{kuehn.stroh.10}, we reduce the original optimization problem 
to a frictionless version, by replacing the mid-price $S_t$ with a suitable ``shadow price'' $\wt{S}_t$. The latter is potentially more favorable for trading but nevertheless leads to 
an equivalent optimal strategy and utility. The key difference to \cite{kuehn.stroh.10} is that  we focus on asymptotic results for small spreads here. Hence, it suffices to determine 
``approximate'' shadow prices: these are at least as favorable, and there exist strategies that trade at the same prices 
in both markets for all spreads, but are only ``almost'' optimal in the frictionless market for 
small spreads. This simplifies the construction significantly, and thereby allows to treat the general framework considered here. 

Indeed, the approximation result established in Lemma \ref{lemma_boundary} suggests that it suffices to look for a frictionless shadow market where the optimal policy oscillates between the upper and lower boundaries $\underline{\beta}_t,\overline{\beta}_t$ at the jump times of the counting processes $N^{(1)}_t,N^{(2)}_t$. To this end, it turns out that one can simply let the shadow price jump to the bid respectively ask price whenever a limit buy respectively sell order is executed, and then let it evolve as the bid respectively ask price until the next jump time. To make this precise, let $\wt{S}_0=(1+\eps_0)S_0$ and define
\beam\label{21.8.2013.1}
\frac{d\wt{S}_t}{\wt{S}_{t-}} & =  &\sigma_t dW_t + 1_{A^{(1)}_t}\left(\left(\frac{1+\varepsilon_t}{1-\varepsilon_t}-1\right)dN^{(2)}_t - \frac1{1-\eps_t}d\eps_t 
- \frac{\sigma_t}{1-\eps_t}d\langle W,\eps\rangle_t\right)\nonumber\\
  & & +1_{A^{(2)}_t}\left(\left(\frac{1-\varepsilon_t}{1+\varepsilon_t}-1\right)dN^{(1)}_t + \frac1{1+\eps_t}d\eps_t 
+ \frac{\sigma_t}{1+\eps_t}d\langle W,\eps\rangle_t\right)\\
& =:&d\widetilde{R}_t,\nonumber
\eeam
with $A^{(1)}$ and $A^{(2)}$ from (\ref{14.5.2013.1}), where the terms $-\frac1{1-\eps_t}d\eps_t - \frac{\sigma_t}{1-\eps_t}d\langle W,\eps\rangle_t$  and
$\frac1{1+\eps_t}d\eps_t + \frac{\sigma_t}{1+\eps_t}d\langle W,\eps\rangle_t$ ensure that $\wt{S}_t=(1-\eps_t)S_t$ on 
$A^{(1)}_{t+}:=\limsup_{n\to\infty} A^{(1)}_{t+1/n}$ and $\wt{S}_t=(1+\eps_t)S_t$
on $A^{(2)}_{t+}:=\limsup_{n\to\infty} A^{(2)}_{t+1/n}$ even for time-varying $\eps_t$. (The correction terms can easily be derived by applying the integration by parts formula to the processes
$1\mp \eps_t$ and $S_t$.) Then:
\beam\label{24.5.2013.1}
(1-\eps_t)S_t \le \wt{S}_t \le (1+\eps_t)S_t,\ \wt{S}_t=(1-\eps_t)S_t\ \mbox{on}\ \Delta N_t^{(1)}>0,\ \wt{S}_t=(1+\eps_t)S_t\ \mbox{on}\ \Delta N_t^{(2)}>0.
\eeam
That is, the frictionless price process $\wt{S}_t$ evolves in the bid-ask spread, and therefore always leads to at least as favorable trading prices for market orders. 
When more favorable trading prices are available due to the execution of limit orders, $\wt{S}_t$ jumps to match these. Hence, trading $\wt{S}_t$ is at least as profitable as 
the original limit order market. 

The key step now is to determine the optimal policy for $\wt{S}_t$. In the corresponding frictionless market, portfolios can be equivalently parametrized directly in terms of monetary positions 
$\widetilde{\eta}_t=\varphi_t\wt{S}_{t-}$ held in the risky asset, with associated wealth process
$$\widetilde{X}^{\widetilde{\eta}}_t = x_0 + \int_0^t \widetilde{\eta}_s d\widetilde{R}_s.$$

We have the following dichotomy:

\begin{Lemma}\label{27.3.2014.3}
Let $(\widetilde{\eta}_t^{\varepsilon})_{\varepsilon\in (0,1)}$ be a uniformly bounded family of policies with associated wealth processes~$\widetilde{X}^{\widetilde{\eta}^\eps}$. 
Then, for every $\delta>0$ there exists an $\varepsilon_\delta>0$ such that:
\beam\label{28.3.2014.1}
P\left(\widetilde{X}^{\widetilde{\eta}^\eps}_t \in (x_0-\delta,x_0+\delta),\ \forall t\in[0,T]\right)\ge 1-\delta\quad \mbox{or}\quad 
E\left[U\left(\widetilde{X}^{\widetilde{\eta}^\eps}_T\right)\right]< U(x_0),\quad\forall \eps\le\eps_\delta. 
\eeam
\end{Lemma}
\begin{proof}
{\em Step 1:} Let $\delta>0$. Strict concavity of $U$ implies
\beam\label{27.3.2014.1}
U(\wt{X}^{\widetilde{\eta}^\varepsilon}_T) & \le & U(x_0)+U'(x_0)(\wt{X}^{\widetilde{\eta}^\varepsilon}-x_0) 
+\left[U(x_0+\delta)-U(x_0)-U'(x_0)\delta\right]1_{\{\wt{X}^{\widetilde{\eta}^\varepsilon}_T\ge x_0+\delta\}}\nonumber\\
& & \qquad\qquad  + \left[U(x_0-\delta)-U(x_0)+U'(x_0)\delta\right]1_{\{\wt{X}^{\widetilde{\eta}^\varepsilon}_T\le x_0-\delta\}}
\eeam
and
\beao
0>\max\left\{U(x_0+\delta)-U(x_0)-U'(x_0)\delta,\ U(x_0-\delta)-U(x_0)+U'(x_0)\delta\right\}=:f(\delta).
\eeao
By Assumption~\ref{22.5.2013.1}, there exists $K\in\bbr_+$ such that $E\left(\wt{X}^{\widetilde{\eta}^\varepsilon}\right) \le x_0+K \varepsilon^{1-\vartheta}$ for all 
$\varepsilon\in(0,1)$. Together with (\ref{27.3.2014.1}), this yields
\beao
E\left[U(\wt{X}^{\widetilde{\eta}^\varepsilon}_T)\right] \le U(x_0)+U'(x_0)K\eps^{1-\vt} + P\left(\widetilde{X}^{\widetilde{\eta}^\eps}_T \not\in (x_0-\delta,x_0+\delta)\right) f(\delta).
\eeao
As a result, the expected utility either lies below $U(x_0)$ or 
\beam\label{26.3.2014}
P\left(\widetilde{X}^{\widetilde{\eta}^\eps}_T \not\in (x_0-\delta,x_0+\delta)\right) \le 
\frac{U'(x_0)K \eps^{1-\vt}}{-f(\delta)}. 
\eeam
Since the right-hand side of (\ref{26.3.2014}) tends to zero as $\eps\to 0$, this already proves the assertion at the terminal time $t=T$.
In the remaining three steps, we show how to extend the assertion to all intermediate times $t \in [0,T]$ in a uniform manner.

{\em Step 2:} Instead of $\wt{X}^{\widetilde{\eta}^\varepsilon}_t$, we first consider the processes $x_0+\int_0^t \widetilde{\eta}_s^{\varepsilon}\sigma_s\,dW_s$, $\varepsilon \in (0,1)$. 
They are true martingales and 
$$
\sup_{\eps\in(0,1)}E\left[\left(\int_0^T \widetilde{\eta}_s^{\varepsilon}\sigma_s\,dW_s\right)^2\right]<\infty.
$$ 
Therefore the family $(|\int_0^T \widetilde{\eta}_s^{\varepsilon}\sigma_s\,dW_s|^p)_{\eps\in(0,1)}$ is uniformly integrable for any $p\in(1,2)$. As a consequence, for every $\xi>0$
there exists a $\delta>0$ such that 
\beam\label{31.3.2014.1}
P\left(\left|\int_0^T \widetilde{\eta}_s^{\varepsilon}\sigma_s\,dW_s\right|\ge\delta\right)\le \delta\quad \implies\quad  
E\left(\left|\int_0^T \widetilde{\eta}_s^{\varepsilon}\sigma_s\,dW_s\right|^p\right)\le \xi
\eeam
for every $\varepsilon \in (0,1)$. By Doob's maximal inequality, 
\beam\label{31.3.2014.2}
E\left(\sup_{t\in[0,T]}\left|\int_0^t \widetilde{\eta}_s^{\varepsilon}\sigma_s\,dW_s\right|^p\right) 
\le \left(\frac{p}{p-1}\right)^p E\left(\left|\int_0^T \widetilde{\eta}_s^{\varepsilon}\sigma_s\,dW_s\right|^p\right).
\eeam
From (\ref{31.3.2014.1}) and (\ref{31.3.2014.2}), we conclude that 
\beao
& & \forall \xi>0\ \exists \delta>0\ \forall \eps\in(0,1)\quad P\left(\left|\int_0^T \widetilde{\eta}_s^{\varepsilon}\sigma_s\,dW_s\right|\ge \delta\right)\le \delta\\
& & \implies\quad P\left(\sup_{t\in[0,T]}\left|\int_0^t \widetilde{\eta}_s^{\varepsilon}\sigma_s\,dW_s\right|\ge \xi\right)\le \xi. 
\eeao

{\em Step 3:} Let us show that the processes $\wt{X}^{\widetilde{\eta}^\varepsilon}_t -x_0 - \int_0^t \widetilde{\eta}_s^{\varepsilon}\sigma_s\,dW_s$ tend to zero uniformly 
in probability for
$\eps\to 0$ (see (\ref{21.8.2013.1}) for the difference of $\wt{R}_t$ and $\int_0^t \sigma_s\,dW_s$). 
By (\ref{label14}), the fact that $\eps_t$ tends linearly to zero, and the uniform boundedness of $\widetilde{\eta}^\eps_t$, the $dN^i_t$- and 
$dt$-terms of $\wt{X}^{\widetilde{\eta}^\varepsilon}_t$ converge to zero in the total variation distance for $\eps\to 0$. 
The same holds for the integrals with respect to $\langle W,\eps\rangle_t$ and the integrals 
with respect to the drift part of $\eps_t$.
To show convergence ``uniformly in probability'' of the integrals with respect to the continuous martingale part of $\eps_t$, we again use the arguments of Step~2.\\

{\em Step 4:} Now, we complete the proof of the lemma by combining the assertions of the previous three steps. 
Let $\xi>0$. By Step~2, there  exist a $\delta\in(0,\xi)$ s.t. for all $\eps\in(0,1)$ the implication
\beam\label{23.1.2015.2}
P\left(\left|\int_0^T \widetilde{\eta}_s^{\varepsilon}\sigma_s\,dW_s\right|\ge \delta\right)\le \delta\quad
\implies\quad P\left(\sup_{t\in[0,T]}\left|\int_0^t \widetilde{\eta}_s^{\varepsilon}\sigma_s\,dW_s\right|\ge \xi/2\right)\le \xi/2
\eeam
holds. 
By Step~3, there exists $\wt{\eps}>0$ s.t.
\beam\label{23.1.2015.3}
P\left(\sup_{t\in[0,T]}\left|\wt{X}^{\widetilde{\eta}^\varepsilon}_s -x_0 - \int_0^t \widetilde{\eta}_s^{\varepsilon}\sigma_s\,dW_s\right|\ge \delta/2\right)
\le \delta/2,\quad \forall \eps\in(0,\wt{\eps}). 
\eeam
In addition, by the triangle inequality, one has
\beam\label{23.1.2015.4}
P\left(\left|\int_0^T \widetilde{\eta}_s^{\varepsilon}\sigma_s\,dW_s\right|\ge \delta\right)
\le P(|\wt{X}^{\widetilde{\eta}^\varepsilon}_T-x_0|\ge \delta/2) +  P\left(\left|\wt{X}^{\widetilde{\eta}^\varepsilon}_T -x_0 
- \int_0^T \widetilde{\eta}_s^{\varepsilon}\sigma_s\,dW_s\right|\ge \delta/2\right).
\eeam
By Step 1, there exists $\wh{\eps}>0$ s.t. for all $\eps\in(0,\wh{\eps})$
\beao
P(|\wt{X}^{\widetilde{\eta}^\varepsilon}_T-x_0|\ge \delta/2) \le \delta/2\quad \mbox{or}\quad E\left[U\left(\widetilde{X}^{\widetilde{\eta}^\eps}_T\right)\right]< U(x_0).
\eeao
Now, let $\eps\in(0,\wt{\eps}\wedge\wh{\eps})$. Either one has $E\left[U\left(\widetilde{X}^{\widetilde{\eta}^\eps}_T\right)\right]< U(x_0)$ or, by (\ref{23.1.2015.4}),
(\ref{23.1.2015.3}), one can apply 
implication~(\ref{23.1.2015.2}) to conclude that $P\left(\sup_{t\in[0,T]}\left|\int_0^t \widetilde{\eta}_s^{\varepsilon}\sigma_s\,dW_s\right|\ge \xi/2\right)\le \xi/2$.
Together with (\ref{23.1.2015.3}), $\delta\le \xi$, and again the triangle inequality, one arrives at 
\beao
P\left(\sup_{t\in[0,T]}\left|\wt{X}^{\widetilde{\eta}^\varepsilon}_s -x_0\right|\ge \xi\right)\le \xi. 
\eeao

\end{proof}

\begin{Bemerkung}
Lemma~\ref{27.3.2014.3} asserts that, for small $\eps$, the wealth process of a policy either remains uniformly close to the initial position or the policy is ``extremely 
bad'' in the sense that the corresponding expected utility is smaller than the obtained by not trading the risky asset at all. 

Starting from an arbitrary uniformly bounded family of policies~$(\widetilde{\eta}^\varepsilon)_{\eps\in(0,1)}$, we may replace $\widetilde{\eta}^\varepsilon$ by $0$ 
for all $\eps$ for which $E[U(\widetilde{X}^{\widetilde{\eta}^\eps}_T)]< U(x_0)$. Then, the modified family of wealth processes performs at least as well as the original one, and Lemma~\ref{27.3.2014.3} implies that the modified family converges to $x_0$ uniformly in probability for $\eps\to 0$. Henceforth, we therefore assume this property already for
$(\wt{X}^{\widetilde{\eta}^\varepsilon})_{\eps\in(0,1)}$ without loss of generality.
\end{Bemerkung}

Let us now compute the expected utility obtained by applying such a family of policies~$(\widetilde{\eta}^\varepsilon)_{\eps\in(0,1)}$.

By Assumption~\ref{22.5.2013.1}, the integrals with respect to $\varepsilon_t$ and $\langle W, \varepsilon \rangle_t$ in \eqref{21.8.2013.1} are dominated 
for $\eps\to 0$. Indeed, the continuous martingale parts are dominated by $\sigma_t\,dW_t$ and the drift terms are dominated by the drifts of the integrals with respect to the 
counting processes $N^i_t$, which are of order $2\eps^{1-\vt}\mathcal{E}_t\lambda^i_t$. Hence, these terms can be safely neglected in the sequel.

It\^o's formula as in 
\cite[Theorem I.4.57]{js.03} and \cite[Theorem II.1.8]{js.03} yield:
\begin{align}
U(\widetilde{X}^{\widetilde{\eta}^\varepsilon}_T)=U(x_0) &+\int_0^T U'(\widetilde{X}^{\widetilde{\eta}^\varepsilon}_{t-})\widetilde{\eta}^\varepsilon_{t} \sigma_t dW_t+\frac{1}{2} \int_0^T U''(\widetilde{X}^{\widetilde{\eta}^\varepsilon}_{t-}) (\widetilde{\eta}^\varepsilon_{t})^2 \sigma^2_t dt  \label{eq:U}\\
&+\left(U(\widetilde{X}^{\widetilde{\eta}^\varepsilon}_{-}+\widetilde{\eta}^\varepsilon x)-U(\widetilde{X}^{\widetilde{\eta}^\varepsilon}_{-})\right)*(\mu^{\widetilde{R}}
-\nu^{\widetilde{R}})_T
+\left(U(\widetilde{X}^{\widetilde{\eta}^\varepsilon}_{-}+\widetilde{\eta}^\varepsilon x)-U(\widetilde{X}^{\widetilde{\eta}^\varepsilon}_{-})\right)*\nu^{\widetilde{R}}_T,\notag
\end{align}
where $\mu^{\widetilde{R}}$ is the jump measure of $\wt{R}$ (see e.g. \cite[Proposition II.1.16]{js.03}) and $\nu^{\widetilde{R}}$ its compensator in the sense 
of \cite[Theorem II.1.8]{js.03}.
The integrals with respect to the Brownian motion $W_t$ and the compensated random measure $\mu^{\widetilde{R}}-\nu^{\widetilde{R}}$ are true martingales. To see this, first consider 
the Brownian integral. By \eqref{eq:marginal} we have 
\beao
U'(\widetilde{X}^{\widetilde{\eta}^\varepsilon}_{t-})\le C \exp(-c_2\widetilde{X}^{\widetilde{\eta}^\varepsilon}_{t-})1_{\{\widetilde{X}^{\widetilde{\eta}^\varepsilon}_{t-}<0\}}
+ U'(0) 1_{\{\widetilde{X}^{\widetilde{\eta}^\varepsilon}_{t-}\ge 0\}}
\eeao 
with constants $C, c_2>0$. Therefore and due to the boundedness of $\widetilde{\eta}^\varepsilon_t \sigma_t$, it suffices to show that
\begin{equation}\label{eq:L2}
E\left[\int_0^T \exp(-2c_2\widetilde{X}^{\widetilde{\eta}^\varepsilon}_t)dt\right]<\infty.
\end{equation}
By the Doleans-Dade exponential formula \cite[Theorem I.4.61]{js.03} and \cite[Theorem II.1.8]{js.03}, we have
\begin{align}
\exp(-2 c_2\widetilde{X}^{\widetilde{\eta}^\varepsilon}_t) =& \exp(-2c_2 x_0) 
\E\left(-2c_2 \int_0^\cdot \widetilde{\eta}_s^\varepsilon \sigma dW_s+\left(\exp(-2c_2 \widetilde{\eta}^\varepsilon x)-1\right)*\left(\mu^{\widetilde{R}}-\nu^{\widetilde{R}}\right)\right)_t \label{eq:DD}\\
&\qquad \times \exp\left(\int_0^t (2 c_2^2 (\widetilde{\eta}^\varepsilon_s)^2 \sigma_s^2)ds +(\exp(-2c_2 \widetilde{\eta}^\varepsilon x)-1)*\nu^{\widetilde{R}}_t\right).\notag
\end{align}
For all $t \in [0,T]$, the ordinary exponential in the above representation is uniformly bounded by a single constant. This is because $\sigma^2_t$ as well as 
$(\widetilde{\eta}^\varepsilon_t)^2$ are both uniformly bounded and, since the intensities $\varepsilon^{-\vartheta}\lambda^{(1)}_t,\varepsilon^{-\vartheta}\lambda^{(2)}_t$ are bounded for any $\varepsilon>0$, the same holds for the jump part (for sufficiently small $\varepsilon$):
\begin{align*}
\left(\exp(-2c_2 \widetilde{\eta}^\varepsilon x)-1\right)*\nu^{\widetilde{R}}_t= &\int_0^t  \left( \exp(-2c_2 \widetilde{\eta}^\varepsilon_s 
2\varepsilon_s/(1-\varepsilon_s))-1\right) 1_{A^{(1)}_s} \varepsilon^{-\vartheta} \lambda^{(2)}_s ds\\
&+\int_0^t \left( \exp(2c_2 \widetilde{\eta}^\eps_s 2\varepsilon_s/(1+\varepsilon_s))-1\right) 1_{A^{(2)}_s} \varepsilon^{-\vartheta} \lambda^{(1)}_s ds.
\end{align*} 
\eqref{eq:L2} now follows since the stochastic exponential in \eqref{eq:DD} is not only a local martingale, but also a supermartingale with decreasing expectation because it is positive for sufficiently small $\varepsilon$.  

The argument for the integral with respect to the compensated random measure $\mu^{\widetilde{R}}-\nu^{\widetilde{R}}$ in \eqref{eq:U} is similar. 
By the mean value theorem, \eqref{eq:marginal}, and \cite[Theorem II.1.33]{js.03} it suffices to show
$$E\left[\exp(-2c_2 \widetilde{X}^{\widetilde{\eta}^\eps}) * \nu^{\widetilde{R}}_T\right] < \infty.$$
But this follows verbatim as for the Brownian integral above, again using that the jumps of $\widetilde{R}_t$ as well as the corresponding jump intensities are all uniformly bounded. 
In summary, \eqref{eq:U} therefore gives\footnote{If the integrals with respect to $\varepsilon_t$ and $\langle W,\varepsilon \rangle_t$ are taken into account explicitly, these only lead to an additional higher-order term that can be bounded by a constant times $\varepsilon \widetilde{\eta}_t^\varepsilon$ for small $\varepsilon$.}
\beam\label{2.9.2013.1}
& & E[U(\widetilde{X}^{\widetilde{\eta}^\varepsilon}_T)]-U(x_0)\nonumber\\
& & =E\left[\int_0^T \frac{1}{2} U''(\widetilde{X}^{\widetilde{\eta}^\varepsilon}_{t-}) (\widetilde{\eta}^\varepsilon_{t})^2  \sigma_{t}^2 dt 
+ (U(\widetilde{X}^{\widetilde{\eta}^\varepsilon}_- + \widetilde{\eta}^\varepsilon x)-U(\widetilde{X}^{\widetilde{\eta}^\varepsilon}_-))*\nu^{\widetilde{R}}_T\right]\nonumber\\
& & \le E\Bigg[\int_0^T \Bigg( \left(\frac{1}{2} U''(\widetilde{X}^{\widetilde{\eta}^\varepsilon}_{t-}) (\widetilde{\eta}_{t}^\varepsilon)^2  \sigma_{t}^2 
+ U'(\widetilde{X}^{\widetilde{\eta}^\varepsilon}_{t-})\widetilde{\eta}_t^\varepsilon \frac{2\varepsilon_t}{1-\varepsilon_t}\varepsilon^{-\vartheta}\lambda^{(2)}_t \right)1_{A^{(1)}_t}\nonumber\\
& & \qquad\qquad\qquad +\left(\frac{1}{2} U''(\widetilde{X}^{\widetilde{\eta}^\varepsilon}_{t-}) (\widetilde{\eta}_{t}^\varepsilon)^2  \sigma_{t}^2 
-U'(\widetilde{X}^{\widetilde{\eta}^\varepsilon}_{t-})\widetilde{\eta}^\varepsilon_t \frac{2\varepsilon_t}{1+\varepsilon_t}\varepsilon^{-\vartheta}\lambda^{(1)}_t\right) 
1_{A^{(2)}_t}\Bigg)dt\Bigg],
\eeam
where for the inequality we have used the concavity of $U$ and inserted the definition of $\widetilde{R}_t$.
By (\ref{eq:marginal}), (\ref{eq:L2}) also shows that the random variables in (\ref{2.9.2013.1}) are integrable. Moreover, for any uniformly bounded family of 
policies~$\wt{\eta}_t^{\varepsilon}$,
\beam\label{20.8.2013.3}
\int_0^T U'(\widetilde{X}^{\wt{\eta}^{\varepsilon}}_t)\,dt\quad\mbox{and}\quad
\int_0^T U''(\widetilde{X}^{\wt{\eta}^{\varepsilon}}_t)\,dt\qquad \mbox{are uniformly integrable for\ }\eps\in(0,\varepsilon_0),
\eeam
where $\varepsilon_0>0$ is a sufficiently small constant. Indeed, it follows from the proof of (\ref{eq:L2}) that the bound therein holds uniformly in $\eps\in(0,\varepsilon_0)$. Then, using Jensen's inequality, we observe that 
$\left(\int_0^T U'(\widetilde{X}^{\wt{\eta}^{\varepsilon}}_t)\,dt\right)^2$, $\left(\int_0^T U''(\widetilde{X}^{\wt{\eta}^{\varepsilon}}_t)\,dt\right)^2$ are uniformly bounded 
in expectation, which in turn yields \eqref{20.8.2013.3}

For fixed wealth $\widetilde{X}^{\widetilde{\eta}^\varepsilon}_{t-}$, the integrand in the upper bound of (\ref{2.9.2013.1})
is a quadratic function in the policy $\widetilde{\eta}^\varepsilon_t$. 
Plugging in the pointwise maximizer~$\frac{2 \varepsilon^{1-\vartheta}\mathcal{E}_t\lambda^{(2)}_t}{\mathrm{ARA}(\widetilde{X}^{\widetilde{\eta}^\varepsilon}_{t-})
\sigma^2_t(1-\varepsilon\mathcal{E}_t)}1_{A^{(1)}_t}
-\frac{2 \varepsilon^{1-\vartheta}\mathcal{E}_t\lambda^{(1)}_t}{\mathrm{ARA}(\widetilde{X}^{\widetilde{\eta}^\varepsilon}_{t-})\sigma^2_t(1+\varepsilon\mathcal{E}_t)}1_{A^{(2)}_t}$,
which is of order $O(\eps^{1-\vartheta})$ (uniformly in $\omega,t$) by \eqref{eq:ARA},
therefore yields the following upper bound:\footnote{If the integrals with respect to $\varepsilon_t$ and $\langle W,\varepsilon \rangle_t$ are taken into account explicitly, this does not change the pointwise optimizer and the corresponding upper bound at the leading order.}
\begin{align*}
&E[U(\widetilde{X}^{\widetilde{\eta}^\varepsilon}_T)]-U(x_0)\\
&\qquad  \leq \int_0^T E\left[ \left(-\frac{U'(\widetilde{X}^{\widetilde{\eta}^\varepsilon}_t)^2}{U''(\widetilde{X}^{\widetilde{\eta}^\varepsilon}_t)}
\frac{2\varepsilon^{2(1-\vartheta)}\mathcal{E}^2_t(\lambda^{(2)}_t)^2}{\sigma^2_t}\right)1_{A^{(1)}_t}
+\left(-\frac{U'(\widetilde{X}^{\widetilde{\eta}^\varepsilon}_t)^2}{U''(\widetilde{X}^{\widetilde{\eta}^\varepsilon}_t)}
\frac{2\varepsilon^{2(1-\vartheta)}\mathcal{E}^2_t(\lambda^{(1)}_t)^2}{\sigma^2_t}\right)
1_{A^{(2)}_t}\right]dt\\
&\qquad \qquad+o(\varepsilon^{2(1-\vartheta)}).
\end{align*}
Here, we used $2\varepsilon_t/(1\mp \varepsilon_t)= \mp 2\varepsilon_t+O(\varepsilon^2)$ and that, by (\ref{20.8.2013.3}), the remainder is uniformly bounded in expectation. 

For the family of feedback policies
\beam\label{2.5.2013.1}
\widetilde{\eta}^{\varepsilon,*}_t=\frac{2 \varepsilon^{1-\vartheta}\mathcal{E}_t\lambda^{(2)}_t}{\mathrm{ARA}(\widetilde{X}^{\widetilde{\eta}^{\varepsilon,*}}_{t-})\sigma^2_t}1_{A^{(1)}_t}
-\frac{2 \varepsilon^{1-\vartheta}\mathcal{E}_t \lambda^{(1)}_t}{\mathrm{ARA}(\widetilde{X}^{\widetilde{\eta}^{\varepsilon,*}}_{t-})\sigma^2_t}1_{A^{(2)}_t}
\eeam
that converges uniformly to zero  as $\eps\to 0$,
this inequality becomes an equality at the leading order $\varepsilon^{2(1-\vartheta)}$, namely:
\begin{align}
&E[U(\widetilde{X}^{\widetilde{\eta}^{\varepsilon,*}}_T)]-U(x_0)\label{eq:Uopt}\\
& \qquad  =  E\Bigg[\int_0^T \Bigg( \left(\frac{1}{2} U''(\widetilde{X}^{\widetilde{\eta}^{\varepsilon,*}}_{t-}) (\widetilde{\eta}_{t}^{\varepsilon,*})^2  \sigma_{t}^2 
+ U'(\widetilde{X}^{\widetilde{\eta}^{\varepsilon,*}}_{t-})\widetilde{\eta}_t^{\varepsilon,*} 
\frac{2\varepsilon_t}{1-\varepsilon_t}\varepsilon^{-\vartheta}\lambda^{(2)}_t \right)1_{A^{(1)}_t}\nonumber\\
&\qquad\qquad\qquad +\left(\frac{1}{2} U''(\widetilde{X}^{\widetilde{\eta}^{\varepsilon,*}}_{t-}) (\widetilde{\eta}_{t}^{\varepsilon,*})^2  \sigma_{t}^2 
-U'(\widetilde{X}^{\widetilde{\eta}^{\varepsilon,*}}_{t-})\widetilde{\eta}^\varepsilon_t \frac{2\varepsilon_t}{1+\varepsilon_t}\varepsilon^{-\vartheta}\lambda^{(1)}_t\right) 
1_{A^{(2)}_t}\Bigg)dt\Bigg]+o(\varepsilon^{2(1-\vartheta)})\nonumber\\
&\qquad  = \int_0^T E\left[ \left(-\frac{U'(\widetilde{X}^{\widetilde{\eta}^{\varepsilon,*}}_t)^2}{U''(\widetilde{X}^{\widetilde{\eta}^{\varepsilon,*}}_t)}
\frac{2\varepsilon^{2(1-\vartheta)}\mathcal{E}^2_t(\lambda^{(2)}_t)^2}{\sigma^2_t}\right)1_{A^{(1)}_t}
+\left(-\frac{U'(\widetilde{X}^{\widetilde{\eta}^{\varepsilon,*}}_t)^2}{U''(\widetilde{X}^{\widetilde{\eta}^{\varepsilon,*}}_t)}
\frac{2\varepsilon^{2(1-\vartheta)}\mathcal{E}_t^2(\lambda^{(1)}_t)^2}{\sigma^2_t}\right)1_{A^{(2)}_t}\right]dt\nonumber\\
&\qquad\qquad\qquad +o(\varepsilon^{2(1-\vartheta)}).\notag
\end{align}
Here, the first equality follows from the mean value theorem because the differential remainder is bounded by 
$C |U'(\widetilde{X}^{\widetilde{\eta}^{\varepsilon,*}}_{t-}+\xi)-U'(\widetilde{X}^{\widetilde{\eta}^{\varepsilon,*}}_{t-})|\varepsilon^{2(1-\eta)}$ for some constant $C>0$, not depending
on $\eps$ as $\widetilde{\eta}^{\varepsilon,*}/\eps^{1-\vartheta}$ is bounded for $\eps\to 0$ by \eqref{eq:ARA}, and some 
bounded random variable $\xi$ which tends to $0$ pointwise for $\eps\to 0$. 
With (\ref{20.8.2013.3}) it follows that the term is uniformly integrable, so that the remainder 
is indeed of order $o(\varepsilon^{2(1-\eta)})$.

As a result:
\begin{equation}\label{eq:Udiff}
E[U(\widetilde{X}^{\widetilde{\eta}^\varepsilon}_T)]-E[U(\widetilde{X}^{\widetilde{\eta}^{\varepsilon,*}}_T)] \leq \varepsilon^{2(1-\vartheta)}M \int_0^T E\left[ \left|\frac{U'(\widetilde{X}^{\widetilde{\eta}^{\varepsilon,*}}_t)^2}{U''(\widetilde{X}^{\widetilde{\eta}^{\varepsilon,*}}_t)}-\frac{U'(\widetilde{X}^{\widetilde{\eta}^\varepsilon}_t)^2}{U''(\widetilde{X}^{\widetilde{\eta}^\varepsilon}_t)}\right|\right]dt +o(\varepsilon^{2(1-\vartheta)}),
\end{equation}
where the constant $M$ is a uniform bound for $2\mathcal{E}_t^2(\lambda^{(2)}_t)^2/\sigma^2_t$ and $2\mathcal{E}_t^2(\lambda^{(1)}_t)^2/\sigma^2_t$. 

% Now, since any family $\widetilde{\eta}^{\varepsilon}_t$ of admissible policies is uniformly bounded and converges to zero pointwise, the dominated convergence theorem 
% for stochastic integrals \cite[Theorem IV.32]{protter.04} shows that $\widetilde{X}_t^{\widetilde{\eta}^{\varepsilon}} \to x_0$ and
With $\widetilde{X}^{\widetilde{\eta}^{\varepsilon}}\to x_0$, we also have   
$U'(\widetilde{X}^{\widetilde{\eta}^{\varepsilon}})^2/U''(\widetilde{X}^{\widetilde{\eta}^{\varepsilon}}) \to U'(x_0)^2/U''(x_0)$ uniformly in probability 
as $\varepsilon \to 0$. As above, by (\ref{20.8.2013.3}) we have uniform integrability, so that this convergence in fact holds in $L^1$. Hence, \eqref{eq:Udiff} and the dominated convergence theorem 
for Lebesgue integrals yield
\beam\label{31.3.2014.3}
E[U(\widetilde{X}^{\widetilde{\eta}^\varepsilon}_T)] \leq E[U(\widetilde{X}^{\widetilde{\eta}^{\varepsilon,*}}_T)] +o(\varepsilon^{2(1-\vartheta)}),
\eeam
that is, the family $(\widetilde{\eta}_t^{\varepsilon,*})_{\varepsilon>0}$ is approximately optimal at the leading order $\varepsilon^{2(1-\vartheta)}$.

Together with \eqref{eq:Uopt}, the same argument also yields that the corresponding leading-order optimal utility is given by 
\begin{align*}
E[U(\widetilde{X}^{\widetilde{\eta}^{\varepsilon,*}}_T)]&= U(x_0) -\frac{U'(x_0)^2}{2U''(x_0)}  E\left[\int_0^T (\widetilde{\eta}^{\varepsilon,*}_t)^2 d\langle \widetilde{R} \rangle_t\right]+o(\varepsilon^{2(1-\vartheta)})\\
&= U\left(x_0+\frac{\varepsilon^{2(1-\vartheta)}}{\mathrm{ARA}(x_0)}  E\left[\int_0^T \left(\frac{2\mathcal{E}_t^2 (\lambda^{(2)}_t)^2}{\sigma_t^2}1_{A^{(1)}_t}+\frac{2\mathcal{E}_t^2 (\lambda^{(1)}_t)^2}{\sigma_t^2}1_{A^{(2)}_t}\right) dt\right]\right)+o(\varepsilon^{2(1-\vartheta)}),
\end{align*}
where the second equality follows from Taylor's theorem and the definition of $\wt{\eta}^{\varepsilon,*}_t$. 

If all the model parameters $\lambda^i_t,\sigma_t,\varepsilon_t$ are constant, the integrals in this formula can be computed explicitly. Indeed, since $P[A^{(1)}_t]=1-P[A^{(2)}_t]=\lambda^{(1)}/(\lambda^{(1)}+\lambda^{(2)})$, it then follows that
$$
E[U(\widetilde{X}^{\widetilde{\eta}^{\varepsilon,*}}_T)]= U\left(x_0+\frac{2\lambda^{(1)} \lambda^{(2)}}{\mathrm{ARA}(x_0) \sigma^2}\varepsilon^{2(1-\vartheta)}T
\right)+o(\varepsilon^{2(1-\vartheta)}).
$$

\subsection{Proof of the Main Result}\label{sec:proof2}

We now complete the proof of our main result. To this end, we use that the policy $\beta_t^\varepsilon$ proposed in Section~\ref{sec:main} is uniformly close to the almost optimal 
policy $\wt{\eta}^{\varepsilon,*}_t$ in the shadow market with price process $\wt{S}_t$ by Lemma \ref{lemma_boundary}. Since trading in the frictionless shadow market is at least 
as favorable as in the original limit order market, and the policy $\beta^\eps_t$ trades at the same prices in both markets, 
this in turn yields the leading-order optimality of $\beta_t^\varepsilon$.

\begin{proof}[Proof of Theorem~\ref{main_theorem}]
Let
\beao
\eta^\eps_t:=\beta^\eps_t\left((1-\eps_t)1_{\{\beta^\eps_t>0\}} + (1+\eps_t)1_{\{\beta^\eps_t<0\}}\right), 
\eeao
where $\beta_t^\eps$ is the solution of \eqref{label1}, which is of order $O(\eps^{1-\vartheta})$ uniformly in $\omega,t$. Note that $\eta_t^\eps$ is the risky position of the policy $\beta_t^\eps$ if the risky asset is valued at the shadow price $\wt{S}_t$ instead of the mid price $S_t$. 
 
{\em Step 1:} We want to compare the $\wt{S}_t$-wealth of $\eta_t^\eps$ with the wealth of the approximate optimizer $\wt{\eta}_t^{\eps,*}$ in the $\wt{S}_t$-market 
defined in (\ref{2.5.2013.1}). 

By (\ref{20.8.2013.3}), in the expansion~(\ref{2.9.2013.1}), one can replace $U'(\wt{X}^{\wt{\eta}^\eps}_t)$ and $U''(\wt{X}^{\wt{\eta}^\eps}_t)$ 
by $U'(x_0)$ and $U''(x_0)$, respectively, leading to a remainder of order $o(\eps^{2(1-\vt)})$,
for any family of policies~$\wt{\eta}_t^\eps$ with the property that $\wt{\eta}_t^\eps/\eps^{1-\vt}$ is uniformly bounded.   
Applied to $\eta^{\varepsilon}_t$ and $\widetilde{\eta}_t^{\varepsilon,*}$, this yields
\beam\label{2.5.2013.2}
& &  E[U(\widetilde{X}^{\eta^{\varepsilon}}_T)]-E[U(\widetilde{X}^{\widetilde{\eta}^{\varepsilon,*}}_T)]\nonumber\\
& &  =  E\Bigg[\int_0^T \Bigg( \left(\frac{1}{2} U''(x_0) (\eta_t^\varepsilon)^2  \sigma_{t}^2 
+ U'(x_0)\eta_t^\varepsilon 
\frac{2\varepsilon_t}{1-\varepsilon_t}\varepsilon^{-\vartheta}\lambda^{(2)}_t \right)1_{A^{(1)}_t}\nonumber\\
& &  \quad\quad +\left(\frac{1}{2} U''(x_0) (\eta_t^\varepsilon)^2  \sigma_{t}^2 
-U'(x_0)\eta^\varepsilon_t \frac{2\varepsilon_t}{1+\varepsilon_t}\varepsilon^{-\vartheta}\lambda^{(1)}_t\right) 
1_{A^{(2)}_t}\Bigg)dt\Bigg]\nonumber\\
& &  \quad\quad - 
E\left[ \int_0^T\left(-\frac{U'(x_0)^2}{U''(x_0)}
\frac{2\varepsilon^{2(1-\vartheta)}\mathcal{E}^2_t(\lambda^{(2)}_t)^2}{\sigma^2_t}\right)1_{A^{(1)}_t}
+\left(-\frac{U'(x_0)^2}{U''(x_0)}
\frac{2\varepsilon^{2(1-\vartheta)}\mathcal{E}^2_t(\lambda^{(1)}_t)^2}{\sigma^2_t}\right)1_{A^{(2)}_t}dt\right]\nonumber\\
& & \quad\quad
+o(\varepsilon^{2(1-\vartheta)})\nonumber\\
& &  = E\Bigg[\int_0^T \frac{1}{2} U''(x_0) \sigma_{t}^2  \left((\eta_t^\varepsilon-\ov{\beta}1_{A^{(1)}_t} - \un{\beta}1_{A^{(2)}_t})^2 
+ (\wt{\eta}_t^{\varepsilon,*}-\ov{\beta}1_{A^{(1)}_t} - \un{\beta}1_{A^{(2)}_t})^2\right) dt\Bigg]\\
& & \quad\quad+o(\varepsilon^{2(1-\vartheta)}).\nonumber
\eeam
By Lemma~\ref{lemma_boundary} and since $\eta_t^\eps-\beta_t^\eps = O(\eps^{2-\vt})$, we have
$(\eta_t^\eps-\ov{\beta}_t 1_{A_t^{(1)}} - \un{\beta}_t 1_{A_t^{(2)}})/\eps^{1-\vartheta}\to 0$ after the first jump of $(N_t^{(1)},N_t^{(2)})$, uniformly in probability. The same holds for 
$\wt{\eta}_t^{\eps,*}$. As the expectation of the first jump time of $(N_t^{(1)},N_t^{(2)})$ is of order $O(\eps^\vartheta)$
(and the integrands in the last line of (\ref{2.5.2013.2}) are uniformly of order $O(\varepsilon^{2(1-\vartheta)})$), this gives
\begin{align*}
E\Bigg[\int_0^T \frac{1}{2} U''(x_0) \sigma_{t}^2  \left((\eta_t^\varepsilon-\ov{\beta}1_{A^{(1)}_t} - \un{\beta}1_{A^{(2)}_t})^2 
+ (\wt{\eta}_t^{\varepsilon,*}-\ov{\beta}1_{A^{(1)}_t} - \un{\beta}1_{A^{(2)}_t})^2\right) dt\Bigg] &= o(\varepsilon^{2(1-\vartheta)}) + O(\varepsilon^{2-\vartheta}) \\
&= o(\varepsilon^{2(1-\vartheta)}).
\end{align*}

{\em Step 2:} Let $(\psi^{0,\eps},\psi^\eps)_{\eps \in (0,1)}$ be an arbitrary admissible family of portfolio processes in the limit order market 
with $(\psi^{0,\eps}_0,\psi^\eps_0)=(x_0,0)$. 
By (\ref{24.5.2013.1}) and Step~1 in the proof of  \cite[Proposition~1]{kuehn.stroh.10}, we have 
\beam\label{24.5.2013.2}
\psi^{0,\eps}_t + \psi^\eps_t 1_{\{\psi^\eps_t\ge 0\}} (1-\eps_t)S_t + \psi^\eps_t 1_{\{\psi^\eps_t< 0\}} (1+\eps_t)S_t \le x_0 + \int_0^t \psi^\eps_s d\wt{S}_s.
\eeam
Due to the boundedness of $\wt{S}/S$, admissibility in the sense of Definition~\ref{21.5.2013.1} implies that 
% the monetary position held in the risky asset is also uniformly bounded when valued at the shadow price. 
the family~$(\psi^\eps\wt{S})_{\eps\in(0,1)}$ is uniformly bounded. Thus, we can apply (\ref{31.3.2014.3}), 
i.e. the family is dominated by the feedback policies~$\wt{\eta}^{\eps,*}$.

Now take the strategy~(\ref{20.8.2013.1}). For the corresponding portfolio process~$(\varphi^{0,\eps}_t,\varphi^\eps_t)_{t \in [0,T]}$  
we have $\eta_t^\eps = \varphi^\eps_t \wt{S}_{t-}$ and -- by construction of the strategy and $\wt{S}_t$ -- \eqref{24.5.2013.2} holds with equality 
for $(\psi_t^{0,\eps},\psi^\eps_t)=(\varphi_t^{0,\eps},\varphi_t^\eps)$. Together with (\ref{2.5.2013.2}) and the approximate optimality 
of $\widetilde{\eta}_t^{\varepsilon,*}$ in the $\wt{S}_t$-market, this yields the assertion. 
\end{proof}

\section{Price Impact of Exogenous Orders}\label{sec:impact}

In our model, bid and ask prices remain unaffected by the execution of incoming orders. This is the most optimistic scenario for liquidity providers because -- modulo inventory risk -- it allows them to earn the full spread between alternating buy and sell trades. Disregarding price impact is reasonable for small noise traders whose orders do not carry any information. For strategic and possibly informed counterparties, however, it is questionable. For these, prices are expected to rise after purchases and fall after sales, respectively (compare, e.g., \cite{glosten.milgrom.85,madhavan.al.97}). Similarly, larger orders of other market participants also move market prices in the same directions by depleting the order book \cite{obizhaeva.wang.13}. 

Our basic model can be extended to incorporate the price impact of incoming orders in reduced form.\footnote{This is similar in spirit to the 
Almgren-Chriss model \cite{almgren.chriss.01} from the optimal execution literature, in that we also do not attempt to specify the dynamics of the whole order book, 
but instead directly model the price moves caused by executions.}  Indeed, suppose that the mid price follows
\beam\label{20.1.2015.1}
\frac{dS_t}{S_{t-}} = \sigma_t\,dW_t -\kappa\eps_t\,dN^{(1)}_t + \kappa\eps_t\,dN^{(2)}_t,
\eeam
for some price impact parameter $\kappa\in[0,1)$. With the information~$\scr{F}_{t-}$, our small investor is allowed to place limit buy orders at the bid price~$(1-\eps_t)S_{t-}$ and
limit sell orders at the ask price~$(1+\eps_t)S_{t-}$ immediately before the jump of $S_t$. However, bid and ask prices jump down after exogenous sell orders arrive at 
the jump times of $N_t^{(1)}$, and jump up after exogenous 
buy orders arrive at the jump times of $N_t^{(2)}$. (Note that this happens irrespective of the liquidity our \emph{small} investor chooses to provide.)
Formally, this means that the self-financing condition~(\ref{eq:sf2}) is replaced by\footnote{For the integrals with respect to $M^B_t$ and $M^S_t$ see  
(\ref{20.1.2015.2}). Since $M^B_t-M^B_{t-}$, $M^S_t-M^S_{t-}$ have to be predictable and since, by the assumptions on $N^{(1)}, N^{(2)}$, the jump times 
of (\ref{20.1.2015.1}) are totally inaccessible stopping times, market orders are actually always executed at $(1\pm\eps_t)S_t$.}
%%% See \cite[Equation (2.2)]{kuehn.stroh.10} for a precise definition of the integrals 
%%% with respect to the processes $M^B$ and $M^S$ that may have double jumps.
\begin{align*}
\varphi^0_t  =  x_0 &-\int_0^t ((1+\eps_s)S_{s-},(1+\eps_s)S_s)\,dM^B_s+\int_0^t ((1-\eps_s)S_{s-}, (1-\eps_s)S_s)\,dM^S_s\\
 &-\int_0^{t-}L^B_s(1-\eps_s)S_{s-}\,dN^{(1)}_s+\int_0^{t-}L^S_s(1+\eps_s)S_{s-}\,dN^{(2)}_s.
\end{align*}
The parameter $\kappa$ represents the fraction of the half-spread $\varepsilon_t S_{t-}$ by which prices are moved.\footnote{Note that, as in \cite{cartea.jaimungal.13}, price impact is permanent here. Tracking an exogenous benchmark in a limit order market with \emph{transient} price impact as in \cite{obizhaeva.wang.13} is studied by \cite{horst.naujokat.13}.} $\kappa=0$ corresponds to the model without price impact studied above. 
Conversely, $\kappa \approx 1$ leads to a model in the spirit of Madhavan et al.\ \cite{madhavan.al.97}, where liquidity providers do not earn the spread, 
but only a small exogenous compensation for their services.\footnote{Indeed, after a successful execution of a limit order the mid price jumps close to the 
limit price of the order if $\kappa \approx 1$. This means that the liquidity provider actually trades at similar prices as in a frictionless market with price process $S_t$. 
If moreover $\alpha^{(1)}_t=\alpha^{(2)}_t$, the mid price is still a martingale and expected profits vanish.} 

In the model, the liquidity provider does not internalize the price impact and therefore continues to post liquidity at the best bid and ask prices. This assumption is made for tractability. Indeed, the motivation for this 
restriction of the considered limit prices is not as compelling as in the basic model with continuous bid-ask prices. Alternatively, as in \cite{cartea.jaimungal.13}, 
one can consider models where the size of the liquidity provider's order is fixed but it can be posted deeper in the order book to mitigate the adverse price impact.\footnote{Also compare \cite{cartea.jaimungal.14}, where two types of models are discusses. In the first one, one can post one limit order for one share with an arbitrary limit price. In the second, limit prices are fixed at the best bid-ask prices, but volume can be arbitrary.} Incorporating strategic decisions concerning order size \emph{and} location in a tractable manner is a challenging direction for future research.

In the above extension of our model, the optimal policy is similar to the one in the baseline version without price impact. One still trades to some position limits
$\underline{\beta}_t,\overline{\beta}_t$ whenever limit orders are executed. However, since the adverse effect of price impact diminishes the incentive to provide liquidity,
$\underline{\beta}_t,\overline{\beta}_t$ are reduced accordingly. If executions move bid and ask prices by a fraction $\kappa$ of the current half-spread~$\varepsilon_t S_{t-}$, then
\begin{equation}\label{eq:boundariesimpact}
\overline{\beta}_t = \frac{2\varepsilon_t ((1-\frac{\kappa}{2})\alpha^{(2)}_t-\frac{\kappa}{2}\alpha^{(1)}_t)}{\mathrm{ARA}(x_0) \sigma^2_t}, \quad \underline{\beta}_t = -\frac{2\varepsilon_t ((1-\frac{\kappa}{2})\alpha^{(1)}_t-\frac{\kappa}{2}\alpha^{(2)}_t)}{\mathrm{ARA}(x_0) \sigma^2_t},
\end{equation}
given that $(1-\frac{\kappa}{2})\alpha^{(2)}_t -\frac{\kappa}{2} \alpha^{(1)}_t$ and $(1-\frac{\kappa}{2})\alpha^{(1)}_t - \frac{\kappa}{2}\alpha^{(2)}_t$ are positive. In the symmetric case
$\alpha^{(1)}_t=\alpha^{(2)}_t=\alpha_t$, i.e., if buy and sell orders arrive at the same rates, this holds if and only if $\kappa<1$. In this case,
$$\overline{\beta}_t=\frac{2\varepsilon_t (1-\kappa)\alpha_t}{2 \mathrm{ARA}(x_0) \sigma^2_t}, \quad \underline{\beta}_t=-\frac{2\varepsilon_t (1-\kappa)\alpha_t}{2 \mathrm{ARA}(x_0) \sigma^2_t},
$$
so that price impact equal to a fraction $\kappa$ of the current half-spread $\varepsilon_t S_{t-}$ simply reduces liquidity provision by a factor of $1-\kappa$. In particular, if $\kappa \approx 1$, then the boundaries can be of order $o(\varepsilon^{1-\vartheta})$. As a result, arrival rates of a higher order than $\varepsilon^{-\vartheta}, \vartheta \in (0,1)$ can be used without implying nontrivial profits as the spread collapses to zero. In any case, the formula \eqref{eq:welfare} for the corresponding leading-order certainty equivalent remains the same after replacing the trading boundaries accordingly.

In addition to reducing the target positions for limit order trades, price impact also alters the rebalancing strategy between these. Recall that price impact increases bid-ask prices 
after the liquidity provider has sold the risky asset, and decreases them after purchases. Hence, immediately starting to trade by market orders to keep the inventory 
in $[\underline{\beta}_t,\overline{\beta}_t]$ is not optimal anymore, since this would more than offset the gains from the previous limit order transactions. 
To circumvent this, one can instead focus solely on limit orders, and ensure admissibility by liquidating the portfolio with market orders and stopping trading altogether if the risky 
position exits the bigger interval~$[2\un{\beta}_t,2\ov{\beta}_t]$. In the limit for small spreads and frequent limit order executions, the probability for this event tends to zero, 
so that the utility loss due to premature liquidation is negligible at the leading order, and the corresponding policy turns out to be optimal.\footnote{The same modification could 
also have been used in the baseline model without price impact. There, however, the exact optimal strategy keeps the inventory between $\underline{\beta}_t,\overline{\beta}_t$ by market 
orders in simple settings \cite{kuehn.stroh.10}, so that we stick to a strategy of that type there.} 

Let us sketch how the arguments from Section \ref{sec:proofs} can be adapted to derive these results. Again, construct a frictionless shadow price process $\widetilde{S}_{t}$, for which the optimal strategy trades at the same prices as in the limit order market. Define
\beao
\wt{S}_t = (1-\eps_t) S_{t-} = \frac{1-\eps_t}{1-\kappa\eps_t}S_t\quad\mbox{if } \Delta N^{(1)}_t>0,
\quad\mbox{and}\quad\wt{S}_t = (1+\eps_t) S_{t-} 
= \frac{1+\eps_t}{1+\kappa\eps_t}S_t\quad\mbox{if } \Delta N^{(2)}_t>0,
\eeao
and assume that the quotient $\wt{S}_t/S_t$ is piecewise constant between the jump time of $N_t^{(1)}, N_t^{(2)}$. These properties are satisfied by the solution of
\beam\label{14.7.2013.1}
\frac{d\wt{S}_t}{\wt{S}_{t-}} & = & \sigma_t dW_t + 1_{A^{(1)}_t}\left(\left(\frac{1+\varepsilon_t}{1-\varepsilon_t}(1-\kappa\eps_t)-1\right)\,dN^{(2)}_t - \kappa\eps_t\,dN^{(1)}_t\right.\nonumber\\
& & \left. + \frac{\kappa-1}{(1-\eps_t)(1-\kappa\eps_t)}\,d\eps_t
+ \frac{(\kappa-1)\kappa}{(1-\eps_t)(1-\kappa\eps_t)^2}\,d\langle \eps,\eps\rangle_t
 + \frac{(\kappa-1)\sigma_t}{(1-\eps_t)(1-\kappa\eps_t)}d\langle W,\eps\rangle_t\right)\nonumber\\
& & \qquad\quad  +1_{A^{(2)}_t}\left(\left(\frac{1-\varepsilon_t}{1+\varepsilon_t}(1+\kappa\eps_t)-1\right)\,dN^{(1)}_t + \kappa\eps_t\,dN^{(2)}_t \right.\nonumber\\
& & \left. + \frac{1-\kappa}{(1+\eps_t)(1+\kappa\eps_t)}\,d\eps_t
+ \frac{(\kappa-1)\kappa}{(1+\eps_t)(1+\kappa\eps_t)^2}\,d\langle\eps,\eps\rangle_t + \frac{(1-\kappa)\sigma_t}{(1+\eps_t)(1+\kappa\eps_t)}d\langle W,\eps\rangle_t\right)
\eeam
with $\wt{S}_0:=(1+\eps_0)(1+\kappa \eps_0)^{-1}S_0$. Here, the terms in the second and fourth line of (\ref{14.7.2013.1}) ensure that $\wt{S}_t$ coincides with 
$(1-\eps_t)(1-\kappa\eps_t)^{-1}S_t$ on $A^{(1)}_{t+}$ and with
$(1+\eps_t)(1+\kappa\eps_t)^{-1}S_t$ on $A^{(2)}_{t+}$. For constant $\eps_t$ these terms disappear. 
As without price impact, they do anyhow not contribute at the leading order for $\eps\to 0$.

This frictionless price process matches the execution prices of limit orders in the original limit order market, as limit orders are executed at their limit prices 
which are fixed {\em before} orders are executed. However, the corresponding jumps due to price impact -- which occur simultaneously with executions in the limit order market -- are only accounted for at the next trade in the frictionless shadow market. Hence, market orders to manage the investor's inventory -- which naturally consist of sales after limit order purchases and vice versa -- can be carried out at \emph{strictly} more favorable price with $\wt{S}_t$. Hence, trading $\wt{S}_t$ is generally strictly more favorable than the original limit order market, and equally favorable only for limit order trades.

As in Section \ref{sec:proof1}, one verifies that a risky position that oscillates between $\underline{\beta}_t,\overline{\beta}_t$ at the jump times of $N_t^{(1)},N_t^{(2)}$ is optimal at the
leading order for $\wt{S}_t$. Similarly as in Section \ref{sec:proof2}, one then checks that the same utility can be obtained in the original limit order market by using the policy proposed
above. Indeed, the corresponding limit order trades are executed at the same prices as for $\wt{S}_t$. For the potential liquidating trade by market orders, there is a single additional
loss of order $O(\eps^{2-\vartheta})=o(\varepsilon^{2(1-\vartheta)})$, which is negligible at the leading order $O(\varepsilon^{2(1-\vartheta)})$. 
The utility lost due to terminating trading early is of order $O(\varepsilon^{2(1-\vartheta)})$, because it is bounded by its counterpart for $\wt{S}_t$, and it follows similarly as 
in the proof of Lemma~\ref{lemma_boundary} that the probability for a premature termination tends to zero as $\varepsilon \to 0$. As a result, the total utility loss due to early 
termination is therefore also not visible at the leading order $O(\varepsilon^{2(1-\vartheta)})$. In summary, the policy proposed above matches the optimal utility in the superior 
frictionless market $\wt{S}_t$ at the leading order, and is therefore optimal at the leading order in the original limit order market as well.

\bibliographystyle{abbrv}
\bibliography{tractrans}

\end{document}